\documentclass[11pt,letterpaper]{article}
\usepackage{graphicx}
\usepackage{latexsym}
\usepackage{amssymb}
\usepackage{amsmath}
\usepackage{amsthm}
\usepackage{MnSymbol}
\usepackage{tikz-cd}
\usepackage{subcaption}
\usepackage[margin=1in]{geometry}
\usepackage{paralist}
\usepackage{nicefrac}
\usepackage[noend]{algpseudocode}
\usepackage{fullpage}
\usepackage{color}
\usepackage{ifpdf}
\usepackage{fullpage}
\ifpdf    
\usepackage[bookmarks=true, pdfpagelabels, plainpages=false]{hyperref}
\else    
\usepackage[hypertex]{hyperref}
\fi
\usepackage{times}

\usepackage{multirow}
\usepackage{nicefrac}

\usepackage[linesnumbered,ruled,vlined]{algorithm2e}

\begin{document}

\newcommand {\ignore} [1] {}

\newtheorem{theorem}{Theorem}[section]
\newtheorem{lemma}[theorem]{Lemma}
\newtheorem{fact}[theorem]{Fact}
\newtheorem{corollary}[theorem]{Corollary}
\newtheorem{definition}{Definition}[section]
\newtheorem{proposition}[theorem]{Proposition}
\newtheorem{observation}[theorem]{Observation}
\newtheorem{claim}[theorem]{Claim}
\newtheorem{assumption}[theorem]{Assumption}
\newtheorem{notation}[theorem]{Notation}
\newtheorem{reduction}{Reduction}

\newcommand{\EXT}{{\texttt{$0$-Extension}}}
\newcommand{\ExtGH}{{\text{Ext}(G,H)}}
\newcommand{\ExtGprimeH}{{\text{Ext}(G',H)}}
\newcommand{\LH}{{\text{Lab}_H}}
\newcommand{\LG}{{\text{Lab}_G}}
\newcommand{\LX}{{\text{Lab}_X}}
\newcommand{\MW}{{\texttt{Multiway Cut}}}
\newcommand{\ML}{{\texttt{Metric Labeling}}}
\newcommand{\VC}{{\texttt{Vertex Cover}}}

\pagenumbering{arabic}

\title{The Metric Relaxation for $0$-Extension Admits an $\Omega(\log^{\nicefrac[]{2}{3}}{k})$ Gap}

\author{
 Roy Schwartz\thanks{The Henry and Marilyn Taub Faculty of Computer Science, Technion. E-mail: \texttt{schwartz@cs.technion.ac.il}.}
 \and
Nitzan Tur\thanks{The Henry and Marilyn Taub Faculty of Computer Science, Technion. E-mail: \texttt{nitzan.tur@cs.technion.ac.il}.}
}

\date{}
\maketitle
\thispagestyle{empty}

\begin{abstract}
We consider the {\EXT} problem, where we are given an undirected graph $\mathcal{G}=(V,E)$ equipped with non-negative edge weights $w:E\rightarrow \mathbb{R}^+$, a collection $ T=\{ t_1,\ldots,t_k\}\subseteq V$ of $k$ special vertices called terminals, and a semi-metric $D$ over $T$.
The goal is to assign every non-terminal vertex to a terminal while minimizing the sum over all edges of the weight of the edge multiplied by the distance in $D$ between the terminals to which the endpoints of the edge are assigned.
{\EXT} admits two known algorithms, achieving approximations of $O(\log{k})$ [C{\u{a}}linescu-Karloff-Rabani {\em SICOMP} '05] and $O(\log{k}/\log{\log{k}})$ [Fakcharoenphol-Harrelson-Rao-Talwar {\em SODA} '03].
Both known algorithms are based on rounding a natural linear programming relaxation called the metric relaxation, in which $D$ is extended from $T$ to the entire of $V$.
The current best known integrality gap for the metric relaxation is $\Omega (\sqrt{\log{k}})$.
In this work we present an improved integrality gap of $\Omega(\log^{\nicefrac[]{2}{3}}k)$ for the metric relaxation.
Our construction is based on the randomized extension of one graph by another, a notion that captures lifts of graphs as a special case and might be of independent interest.
Inspired by algebraic topology, our analysis of the gap instance is based on proving no continuous section (in the topological sense) exists in the randomized extension.
\end{abstract}

\maketitle

\section{Introduction}
We consider the {\EXT} problem, where we are given an undirected graph $\mathcal{G}=(V,E)$, equipped with non-negative weights $w:E\rightarrow \mathbb{R}^+$ on the edges, a set of $k$ special vertices $T=\{ t_1,\ldots,t_k\}\subseteq V$ called terminals, and a semi-metric $D:T\times T \rightarrow \mathbb{R}^+$ over the terminals.
The goal is to partition the vertices of $\mathcal{G}$ into $k$ parts $\{ S_1,\ldots,S_k\}$, where $t_i\in S_i$ for every $ i=1,\ldots,k$, while minimizing the total cost of the partition.
Given a partition, each edge contributes to the cost the distance in the semi-metric $D$ between the terminals that represent the part each of its two endpoints belong to.
Formally, the goal is to find a labeling $f:V\rightarrow T$, where $f(t_i)=t_i$ for every $i=1,\ldots,k$, that minimizes:
$$ \sum _{(u,v)\in E}w_e \cdot D(f(u),f(v)).$$

The {\EXT} problem was first considered by Karzanov \cite{K98}, and it takes its name from the fact that the objective is to extend $D$ from $T$ to a semi-metric on the entire of $V$ subject to the constraint that every non-terminal vertex is required to be at distance $0$ from one of the terminals.
{\EXT} captures the classic {\MW} problem \cite{AMM17,BCKM19,BNS18,BSW17,BSW19,CKR00,CT06,DJPSY94,FK00,KKSTY04,SV13} for the special case where $D$ is a uniform metric, {\em i.e.}, $D(t_i,t_j)=1$ for all $i\neq j$ and $0$ otherwise.
The fact that the cost of an edge depends on the terminals to which its endpoints are assigned to, makes {\EXT} considerably more challenging than {\MW}.
Moreover, {\EXT} is also a special case of the {\ML} problem \cite{AFHKTT04,CKNZ04,CN07,GT00,KKMR09,KT02}, where there are no terminals but each vertex is associated with an assignment cost to each of the terminals.
It should also be noted that {\EXT}, alongside the algorithmic techniques developed to tackle it, are related to other topics, {\em e.g.}, Lipschitz extension in Banach spaces \cite{JLS86,LN04}, metric embedding \cite{KLMN05}, and the approximation of metrics by tree metrics \cite{FRT07}.

Since its introduction, many special cases of {\EXT} were studied, {\em e.g.}, planar and minor free graphs, bounded diameter $D$, and $D$ is the shortest path metric of a high girth expander graph.
When considering the general case, C{\u{a}}linescu, Karloff and Rabani \cite{CKR05} were the first to provide an approximation achieving a guarantee of $O(\log{k})$.
This bound was later improved by Fakcharoenphol, Harrelson, Rao, and Talwar \cite{FHRT03} to $O(\log{k}/\log{\log{k}})$ and is the current best known approximation for the general case.
Both these algorithms are based on a natural linear programming relaxation called the {\em metric relaxation}, which optimizes over all metric extensions of $D$ to the entire of $V$.
This metric relaxation was first given in \cite{K98}.
When considering lower bounds for the metric relaxation, an integrality gap of $\Omega(\sqrt{\log{k}})$ is known and was given by \cite{CKR05}.
This is the current best known integrality gap for the metric relaxation.

A different relaxation was given by Chekuri, Khanna, Naor and Zosin \cite{CKNZ04}, in which every vertex corresponds to a distribution over terminals and the distance between the distributions of two neighboring vertices is measured using earthmover distances. 
This relaxation is known as the {\em earthmover relaxation} for {\EXT}.
The current best known integrality gap for this relaxation also equals $\Omega(\sqrt{\log{k}})$ and was given by Karloff, Khot, Mehta and Rabani \cite{KKMR09} (one should note that the gap example of \cite{KKMR09} differs from that of \cite{CKR05}).
Moreover, it was proved by Manokaran, Naor, Raghavendra and Schwartz \cite{MNRS08} that assuming the unique games conjecture any integrality gap of the earthmover relaxation translates to hardness of the same value as the gap.
Since the earthmover relaxation finds the best transportation metric that extends $D$ to the entire of $V$, as opposed to the metric relaxation that finds the best (arbitrary) metric that extends $D$ to the entire of $V$, one can infer that the earthmover relaxation is at least as strong as the metric relaxation for {\EXT}.

Despite the above, which suggests that one should always favor the earthmover relaxation over the metric relaxation when focusing on {\EXT}, there are two important things to note.
First, both known algorithms for the general case of {\EXT} \cite{CKR05,FHRT03} are based on the metric relaxation and not the earthmover relaxation.
Hence, it is not known how to algorithmically exploit the fact that the metric in the relaxation is a transportation metric (as in the earthmover relaxation) as opposed to an arbitrary metric (as in the metric relaxation).
Second, both integrality gap instances of both relaxations equal $\Omega(\sqrt{\log{k}})$, a barrier that seems inherent in known approaches for designing integrality gap instances for {\EXT} (see also Section \ref{sec:Techniques}).
The above lead to the following questions regarding {\EXT}, that were raised in \cite{FHRT03}: can the gap between $O(\log{k}/\log{\log{k}})$ and $\Omega(\sqrt{\log{k}})$ for the metric relaxation be closed?
Is the earthmover relaxation indeed strictly stronger than the metric relaxation and what is its tight integrality gap?
Unfortunately, no progress has been made regarding these questions since the above mentioned works \cite{CKR05,FHRT03,KKMR09}.

\subsection{Our Result}\label{sec:Results}
In this work we make progress in answering the above questions of \cite{FHRT03}, and present an improved integrality gap of $\Omega(\log^{\nicefrac[]{2}{3}}k)$ for the metric relaxation for {\EXT}.
This improves the previous known gap of $\Omega(\sqrt{\log{k}})$, given by \cite{CKR05}.
The following theorem summarizes our main result.
\begin{theorem}\label{thrm:Main}
For every $k$ the metric relaxation for {\EXT} admits an integrality gap of $\Omega(\log^{\nicefrac[]{2}{3}}k)$.
\end{theorem}

\subsection{Our Approach}\label{sec:Techniques}
To better present our approach for designing an improved integrality gap, we start with an intuitive bird's-eye description of how both algorithms for {\EXT} operate \cite{CKR05,FHRT03}.
In what follows it is assumed that distances are partitioned by a logarithmic scale, and thus distances within the same scale are equal up to a multiplicative constant.
First, every non-terminal vertex $u$ chooses a scale that is comparable to the distance of $u$ to its closest terminal (a size that is denoted by $A_u$).
We note that this step must be randomized, such that a pair of close neighboring non-terminals should be in the same scale with a high enough probability.
Second, $u$ is randomly assigned to a terminal that is within a distance of (roughly) the chosen scale of $A_u$ from it. 
This is done independently for each scale.
The algorithm of \cite{CKR05} introduced the above approach and achieves an approximation of $O(\log{k})$.
The improvement to $O(\log{k}/\log{\log{k}})$ \cite{FHRT03} is obtained by observing that the above choice of scale is not the only possible one and choosing any scale in a wide range also provides the same worst case approximation guarantee of $O(\log{k})$.
Thus, \cite{FHRT03} exploited this observation and proved that a smart random choice of a scale suffices to obtain the improved result.

An important conclusion of the above bird's-eye description of how both known algorithms operate, is that a hard instance for {\EXT} should be hard in a {\em wide range} of scales.
A fact that is crucial when designing and analyzing integrality gap instances.
For example, both integrality gap instances \cite{CKR05,KKMR09} are based in their core on a {\em single} expander graph, which in turn (after some modifications) cannot produce a gap larger than $\Omega(\sqrt{\log{k}})$.
Intuitively, the reason for this is the following.
Fix an arbitrary scale.
If the scale is at most $O(\sqrt{\log{k}})$ then the algorithmic approach of \cite{CKR05} yields a loss of $O(\sqrt{\log{k}})$.
Otherwise, if the scale is $\Omega(\sqrt{\log{k}})$ one can always use instead a scale that is $\Theta(\log{k})$, which is the expander's diamater.
In this scale one can easily solve the problem with only a constant loss, {\em e.g.}, by assigning all non-terminals to the same terminal.
As before, there is an overall loss of $O(\sqrt{\log{k}})$ that originates from shifting to a larger scale.
Hence, the $\Omega(\sqrt{\log{k}})$ barrier seems to be inherent in known approaches for designing integrality gap instances for {\EXT}.

To ensure that our instance is hard enough in multiple scales, our improved $\Omega(\log^{\nicefrac[]{2}{3}}{k})$ integrality gap instance is based on the natural notion of the randomized extension of one graph by another.
Specifically, given two graphs $G$ and $H$, the randomized extension of $G$ by $H$ is constructed by placing a copy of $H$ for every vertex of $G$ (such a copy of $H$ is called a ``cloud'') and two neighboring (with respect to $G$) copies of $H$ are connected by a uniform random perfect matching.
This type of extension captures {\em  lifts of graphs} as a special case where $H$ contains no edges.
Lifts of graphs have attracted much attention in recent years \cite{AKM13,AKKSTV08,BL06,F03,LP10,MSS13,MO20,MOP20,OW20,RSW06} due to their applications to the construction of exapnder graphs and their relation to the unique games conjecture.
To the best of our knowledge, we are not aware of any prior use of the more general notion of the randomized extension of $G$ by $H$.
Moreover, it should be noted that the randomized extension of $G$ by $H$ is reminiscent of group extension, hence its name.
We believe this notion might be of independent interest.
In our instance, $G$ and $H$ will be equipped each with a uniform length function over the edges such that each length function corresponds to a different scale.
This enables us when considering the random extension of $G$ by $H$ to obtain better hardness for a wider range of scales.

For every possible realization of the randomized extension of $G$ by $H$, we construct an instance for \EXT~such that a small integrality gap implies a special property of the extension which we call a ``split''.
This property is a (much) weaker version of the property of having a representative vertex in each cloud such that neighboring ``clouds'' have close representatives.
It is worth mentioning that the definition of ``split'' is inspired by the notion of split extensions of groups when the graphs $G$ and $H$ are Cayley graphs. 

Before giving a more detailed overview of the proof, we only mention that our proof has two parts.
The first is that a cheap integral solution to our instance implies a ``split''.
The second is that a ``split'' does not occur with a positive probability (over the random choice of the extension of $G$ by $H$).
Thus, we can conclude that there is an instance for which there is no cheap integral solution.
This instance is the gap instance.
The intuition for our analysis of integral solutions of our gap instance, as well as the exact definition of ``split'', comes from algebraic topology.
However, the proof is (almost) self contained and no knowledge in algebraic topology is assumed.

\subsection{Proof Overview}

On a high level, the construction of the instance takes two copies of our randomized extension of $G$ by $H$ and uses one copy as terminals.
Thus, one can think of an integral solution as a map from the extension to itself.
We further prove that for most of the ``clouds'' this map must be almost a constant (when restricted to the ``cloud''). 
This induces the representative of a ``cloud''.
We additionally prove that neighbouring ``clouds'' have close representatives.

The only obstruction for this map to give us the desired ``split'' is that the image of a cloud may not lie inside the cloud.
Fortunately, while the latter may be true, one can use the high girth of $G$ to prove that the induced map from $G$ to itself (which takes a ``cloud'' to its representative's ``cloud'') is close enough (in some topological sense) to the identity.
This notion of closeness is inspired from algebraic topology: we essentially prove that the map is homotopic to the identity.
Thus, persevering the homologies of $G$. As we do not wish to assume any knowledge in topology, we instead directly prove that the map preserves the cycle structure of $G$ (which is in fact the first homology), {\em i.e.}, every cycle of $G$ is mapped, up to shrinkage, to itself.

To finalize the proof, we show that with high probability such a ``split'' does not exist.
We use a union bound argument, however there are too many possible realizations of the extension of $G$ by $H$.
Therefore, we introduce a combinatorial structure we call a certificate, which can be intuitively thought of as a description of a subgraph of a realization of the extension of $G$ by $H$ which contains all shortest paths between neighbouring representatives.
Building upon the cycle structure (mentioned above) and by using linear algebra, we show that this subgraph's cycle structure is at least as rich as the cycle structure of $G$ itself.
This idea comes again from algebraic topology as we essentially bound the Euler characteristic of the subgraph.
We conclude the union bound by proving that each cycle in the subgraph will give us a constraint which holds with low probability and by bounding the number of certificates.

\subsection{Related Work}\label{sec:RelatedWork}
Improved approximation guarantees for some special cases of \\{\EXT} were studied: $O(1)$ if the graph is planar \cite{CKR05}; \\$O(\sqrt{\text{diam}(D)})$ where $\text{diam}(D)$ is the ratio between the largest and smallest distance in $D$ and $O(1)$ if $D$ is the shortest path metric of a high girth expander \cite{KKMR09}.
Recalling that {\EXT} captures {\MW} as a special case, the fact that the latter is known to be APX-hard \cite{DJPSY94} implies that the former is APX-hard as well.
As previously mentioned, the works of \cite{KKMR09} and \cite{MNRS08} imply that assuming the unqiue games conjecture {\EXT} admits a hardness of $\Omega (\sqrt{\log{k}})$.
An additional and incomparable hardness of $\Omega (\log^{(\nicefrac[]{1}{4}-\varepsilon)}{n})$, for any $\varepsilon >0$, was given by \cite{KKMR09}, assuming $\text{NP}\nsubseteq \text{DTIME}(n^{\text{polylog}n})$.

Focusing on {\MW}, \cite{DJPSY94} introduced the problem and proved that a simple greedy algorithm provides an approximation of $2(1-\nicefrac[]{1}{k})$.
This was improved by \cite{CKR00} who suggested a geometric relaxation and used it to achieve an approximation of $(\nicefrac[]{3}{2}-\nicefrac[]{1}{k})$.\footnote{We note that the geometric relaxation of \cite{CKR00} for {\MW} is the earthmover relaxation when restricted to the case where $D$ is a uniform metric.}
A sequence of works, all based on the above geometric relaxation, provided improved approximations \cite{BNS18,KKSTY04,SV13}, culminating in an approximation of $1.2965$ \cite{SV13}.
Similarly to {\EXT}, the work of \cite{MNRS08} implies that any integrality gap for the geometric relaxation translates into hardness of the same value as the gap assuming the unique games conjecture.
\cite{FK00} provided an integrality gap of $8/(7+1/(k-1))$, which was subsequently improved to $1.2$ by \cite{AMM17} and to $1.20016$ by \cite{BCKM19}.
Special cases and variants of {\MW} were also studied.
For $k=3$ a tight approximation with a matching lower of $\nicefrac[]{12}{11}$ was given by \cite{CT06,KKSTY04}, whereas improved approximations for $k=4,5$ were provided by \cite{KKSTY04}.
In the case the graph is dense and unweighted \cite{AKK99,FK96} provided polynomial time approximation schemes.
The node variant of {\MW} was studied by \cite{GVY04}, who presented an approximation of $2(1-\nicefrac[]{1}{k})$ and proved that improving the $2$ factor in their approximation yields an approximation better than $2$ for {\VC}, which assuming the unique games conjecture is impossible \cite{KR08}.
Additionally, the directed variant of {\MW} was studied by \cite{NZ01}, who presented an approximation of $2$, improving the previous approximation of $O(\log{k})$ \cite{GVY04}.

Focusing on {\ML}, \cite{KT02} introduced the problem and presented an approximation of $2$ for the case $D$ is a uniform metric and $O(\log{k})$ for general metrics (the latter is based on the approximation of metrics by tree metrics \cite{B96,B98,FRT07}).
Using the earthmover relaxation of \cite{CKNZ04}, \cite{AFHKTT04} presented an approximation of $O(\log{n})$.
An integrality gap of $\Omega (\log{k})$ for the earthmover relaxation was given by \cite{KKMR09}, which in conjunction with  \cite{MNRS08} also translates into hardness of the same value assuming the unique games conjecture.
An additional and incomparable hardness result of $\Omega (\log^{(\nicefrac[]{1}{2}-\varepsilon)}k)$, for every $\varepsilon >0$, assuming $\text{NP}\nsubseteq \text{DTIME}(n^{\text{polylog}n})$ was given by \cite{CN07}.
Special cases of {\ML} also admit improved approximations:
$O(1)$ if $D$ is a planar metric \cite{AFHKTT04}; $4$ if $D$ is a truncated linear metric \cite{GT00} (this was improved to $2+\sqrt{2}$ by \cite{CKNZ04}); and an exact solution if $D$ is a linear metric \cite{CKNZ04}.

Lifts of graphs have been studied extensively in recent years.
The main motivation for studying lifts comes from the generation of random {\em regular} expander graphs \cite{AKM13,BL06,LP10,RSW06}, and specifically Ramanujan graphs \cite{F03,MSS13,MO20,MOP20,OW20}.
Lifts of graphs are related to additional topics, {\em e.g.}, the unique games conjecture \cite{AKKSTV08}.

\noindent {\bf{Paper Organization.~~~~~}}
Section \ref{sec:Prelim} contains the formal definition of the metric relaxation for {\EXT}.
Section \ref{sec:GapInstance} defines the notion of randomized extension of one graph by another and our integrality gap instance.
Section \ref{sec:FracSol} is dedicated to analyzing the fractional solution, whereas Section \ref{sec:IntegralSol} contains the basic notions we introduce towards analyzing integral solutions: cycle-homeomorphism and split.
Section \ref{sec:NoSplit}, building upon Section \ref{sec:IntegralSol}, introduces the concept of a certificate and utilizes it to finalize the analysis of integral solutions.

\section{Preliminaries}\label{sec:Prelim}
Recall that a semi-metric space $(V,\delta)$ is comprised of a ground set $V$, and a semi-metric function $\delta:V\times V\rightarrow \mathbb{R}^+$ satisfying: $(1)$ $\delta (u,u)=0$ for every $u\in V$; $(2)$ $\delta (u,v)=\delta (v,u)$ for every $u,v\in V$; and $(3)$ $\delta (u,v)+\delta(v,w)\geq \delta (u,w)$ for every $u,v,w\in V$.\footnote{We note that if condition $(1)$ is changed to $\delta (u,v)=0$ if and only if $u=v$ then $(V,\delta)$ is a metric space.}
The metric relaxation for {\EXT} is denoted by $(MET)$ and is defined as follows:
\begin{align}
(MET)~~~~~\min ~~~& \sum _{e=(u,v)\in E} w_e \cdot \delta(u,v) & \nonumber \\
s.t. ~~~& \text{$(V,\delta)$ is a semi-metric space} &\\
& \delta(t_i,t_j)=D(t_i,t_j) & \forall t_i,t_j\in T, i\neq j
\end{align}
Clearly, $(MET)$ can be formulated as a linear program.
Additionally, in the paper we denote by $\log$ the natural logarithm.

\section{Integrality Gap Instance}\label{sec:GapInstance}
In order to present our integrality gap instance, we introduce the notion of {\em randomized extension of $G$ by $H$}, for given two graphs $G$ and $H$.
Informally, given $G$ and $H$, the randomized extension of $G$ by $H$ is a random graph whose vertices are obtained by inflating every vertex of $G$ into a copy of $H$.
Each inflated vertex of $G$ is called a {\em cloud}, and two clouds (which correspond to neighboring vertices of $G$) are connected by a uniformly random perfect bipartite matching.

As previously mentioned in Section \ref{sec:Techniques}, the above definition is inspired by the groups extensions.
Specifically, if $G$ and $H$ are Cayley graphs of two groups (which for simplicity of presentation we also denote by $G$ and $H$), then {$\ExtGH$} is a distribution over graphs that contains in its support the Cayley graphs of all group extensions of the group $G$ by the group $H$.

The following definition formally introduces the above notion and extends it to the case where $G$ and $H$ are equipped with edge lengths.
One can view the definition in Figure \ref{fig:ExtGH}.

\begin{definition}\label{def:SemiDirect}
Given two graphs $G=(V_G,E_G)$ and $H=(V_H,E_H)$ denote by ${\ExtGH}$ the {\em randomized extension of $G$ by $H$} which is the following distribution over graphs whose vertex set is $V_G\times V_H$ and edge set $E_{\ExtGH}$ is sampled in the following manner:
\begin{enumerate}
    \item For every $g\in V_G$: $((g,h_1),(g,h_2))\in E_{{\ExtGH}}$ if and only if $(h_1,h_2)\in E_H$ (intra-cloud edges).
    
    \item For every $(g_1,g_2)\in E_G$ add to $ E_{{\ExtGH}}$ a uniformly random perfect bipartite matching between the following two sets of vertices of $V_{{\ExtGH}}$: $\{ (g_1,h):h\in V_H\}$ and $ \{ (g_2,h):h\in V_H\}$ (inter-cloud edges). 
\end{enumerate}
Moreover, if $\ell_G:E_G\rightarrow \mathbb{R}^+$ and $\ell_H:E_H\to \mathbb{R}^+$ are non-negative lengths on the edges of $G$ and $H$ respectively, then ${\ExtGH}$ is equipped with the following length function $\ell _{{\ExtGH}}$:
\begin{enumerate}
    \item $\ell _{{\ExtGH}}(((g,h_1),(g,h_2)))=\ell _H ((h_1,h_2))$ for every\\ $ ((g,h_1),(g,h_2))\in E_{{\ExtGH}}$ (intra-cloud lengths).
    \item $\ell _{{\ExtGH}}(((g_1,h_1),(g_2,h_2)))=\ell _G ((g_1,g_2))$ for every\\ $ ((g_1,h_1),(g_2,h_2))\in E_{{\ExtGH}}$ where $g_1\neq g_2$ (inter-cloud lengths).
\end{enumerate}
\end{definition}

\begin{figure}[t]
  \centering
    \includegraphics[width=0.7\linewidth]{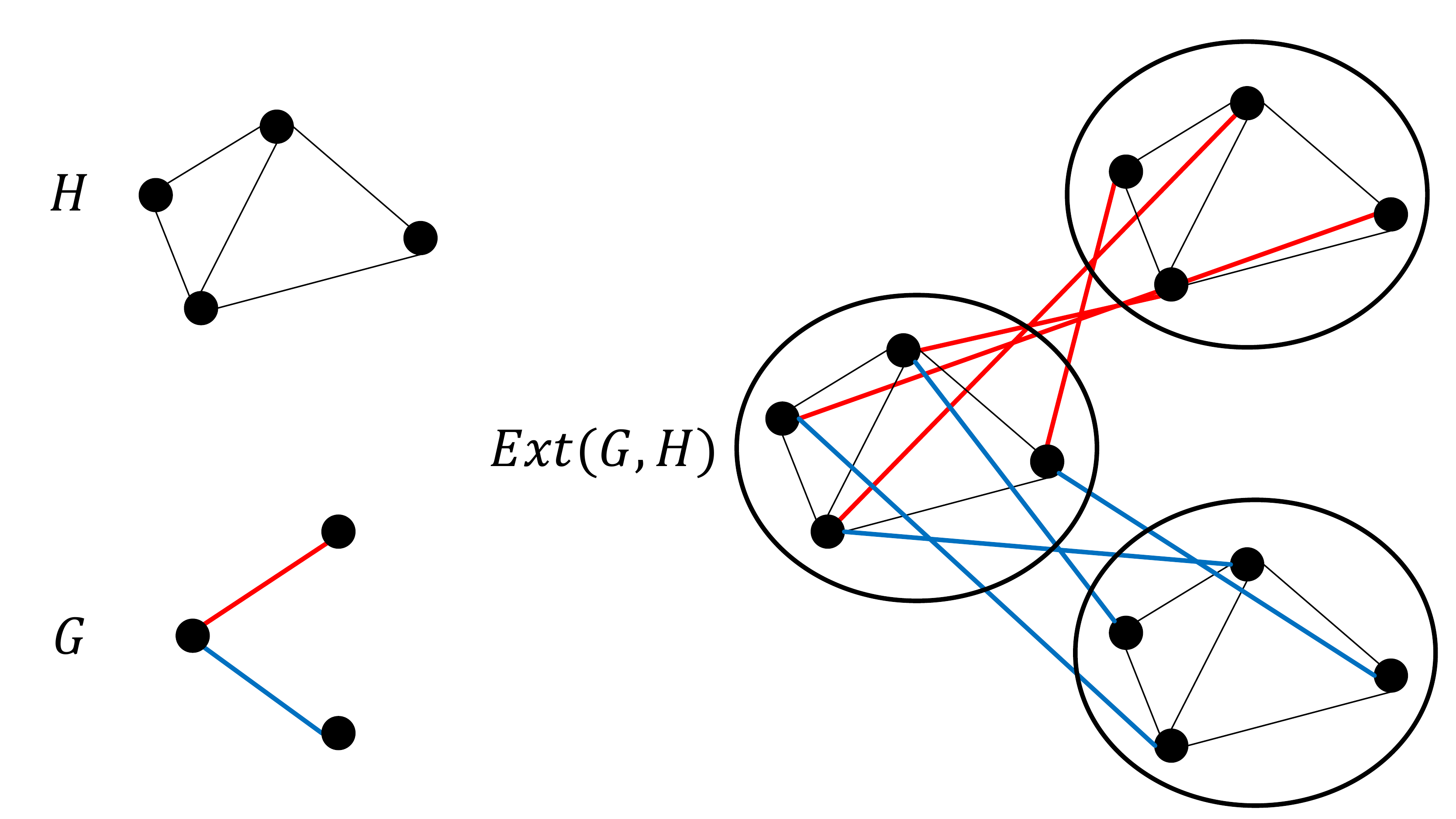}
    \caption{A random sample of $\ExtGH$}
    \label{fig:ExtGH}
  \end{figure}
  
  \begin{figure}[t]
    \includegraphics[width=0.7\linewidth]{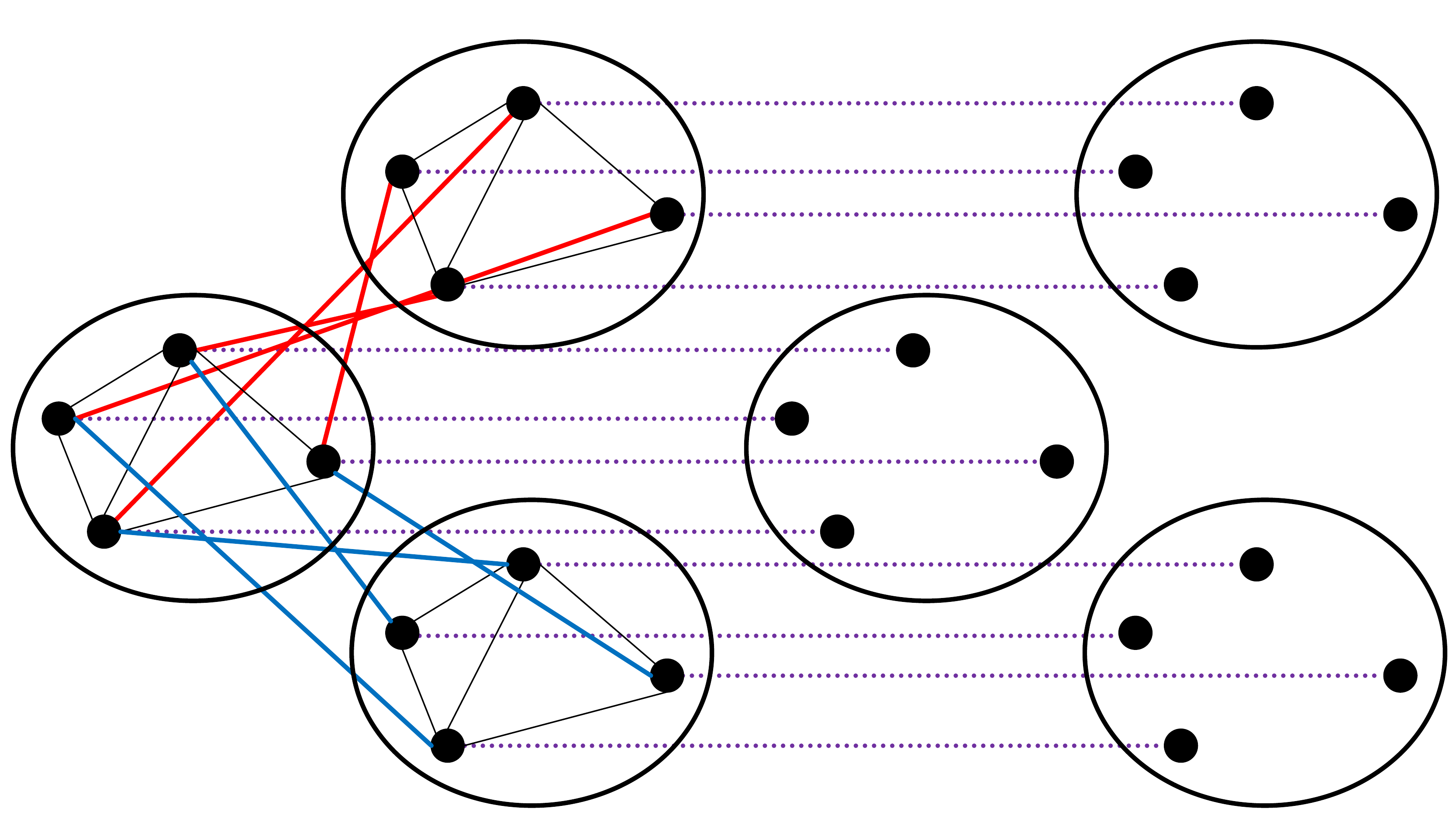}
    \caption{The corresponding instance of $\mathcal{I}(G_X,\ell,L)$}
    \label{fig:Inst}
\end{figure}

Let us now focus on our integrality gap instance for {\EXT}.
Our construction is parameterized by a weighted graph $G_X=(X,E_X)$ (over some collection of points $X$), equipped with non-negative edge lengths $\ell:E_X\rightarrow \mathbb{R}^+$, and a non-negative number $L$.
We denote this instance by $\mathcal{I}(G_X,\ell,L)$.

In what follows, we formally define $\mathcal{I}(G_X,\ell,L)$, given $G_X$, $\ell$ and $L$.
The graph $\mathcal{G}=(V,E)$ has vertices $V\triangleq X \cupdot T$ (where $T$ is a disjoint copy of $X$).
$T\subseteq V$ is set to be the terminals.
For simplicity of presentation, we use the notation $v$ for a point in $X$ and $v_T$ for its counterpart in $T$.
The metric $D$ on $T$ is defined as follows, where $D_X:X\times X\rightarrow \mathbb{R}^+$ is the shortest path metric of the given graph $G_X$ equipped with the given length function $\ell$:
$$ D(u_T,v_T)\triangleq \begin{cases} D_X(u,v)+2L & \text{if }u\neq v\\ 0 & \text{if }u=v.\end{cases}$$
All that remains is to define the edges $E$ of $\mathcal{G}$ and the weight function $w$ over $E$.
An edge $e=(u,v)$ is added to $E$ with weight $w_e=1/\ell(u,v)$ for every $(u,v)\in E_X$, and an edge $e=(v,v_T)$ is added to $E$ with weight $w_e=1/L$ for every $v\in X$.

There are two things to note, regarding the above instance definition.
First, one can easily verify that $D$ is a metric space over $T$.
The reason is that $D$ equals $D_X$ with an additive term of $2L$.
Second, $D$ can be seen as the shortest path metric over $\mathcal{G}$, when restricted only to distances between pairs of points in $T$, where $\mathcal{G}$ is equipped with the following edge lengths: $(u,v)\in E_X$ whose length is $D_X(u,v)$ and $(v,v_T)\in E$ whose length is $L$.

In order to conclude the construction of our instance, we are required to choose $G_X$, $\ell$, and $L$.
Let $G=(V_G,E_G)$ and $H=(V_H,E_H)$ be two graphs equipped with non-negative edge lengths $\ell_G:E_G\rightarrow \mathbb{R}^+$ and $\ell_H:E_H\rightarrow \mathbb{R}^+$, respectively, such that:
$(1)$ $G$ and $H$ are both expanders each with $n$ vertices and a constant bounded degree $d \geq 3$ (to be determined later); $(2)$ $G$ has girth $\Omega(\log{n})$; $(3)$ $\ell _G(e)\equiv \log^{\nicefrac[]{2}{3}}n$ for every $e\in E_G$; and $(4)$ $\ell _H(e)\equiv \log^{\nicefrac[]{1}{3}}n$ for every $e\in E_H$.
We choose $G_X$ to be an edge weighted random graph sampled from {$\ExtGH$}, and $\ell$ to be $ \ell _{{\ExtGH}}$, according to Definition \ref{def:SemiDirect} applied to the above $G$, $H$, $\ell_G$ and $\ell_H$ (thus $X=V_G\times V_H$ and $ E_X=E_{\ExtGH}$).
Finalizing the construction we set $L=\log{n}$.
One can view the above construction in Figure \ref{fig:Inst}.

For convenience of analysis we assume that both $G$ and $H$ are a Cayley graph of some group with respect to $d$ generators (recall that the degree of a Cayley graph is the number of generators).
For example, one can choose the group appearing in \cite{LPS88} which satisfies all the above properties.
In what follows we focus on $H$, but the discussion applies also to $G$.
Specifically, given an edge, it has two labels which correspond to two operations: if $e=(u,v)\in E_H$ then $(u,e)$ has label $x$ if $u\cdot x=v$ in the group (and thus $(v,e)$ has label $x^{-1}$ in the group since $v\cdot x^{-1}=u$).
Moreover, every vertex in $H$ is uniquely identified with an element of the group and vice versa.

It is important to note that our construction is probabilistic.
Therefore, in order to prove our main result, Theorem \ref{thrm:Main}, we prove that with a positive probability that instance satisfies several properties that imply it has a sufficiently large integrality gap.
Not surprisingly, the bulk of the analysis is dedicated to analyzing the integral solution. 
Moreover, since $k=|X|=n^2$ in $\mathcal{I}(G_X,\ell,L)$, in the remainder of the paper for simplicity of presentation all guarantees are stated with respect to $n$ and not $k$.

\section{The Fractional Solution}\label{sec:FracSol}

In this section we upper bound the value of the optimal fractional solution of the relaxation $ (MET)$ for $\mathcal{I}(G_X,\ell,L)$.
The following lemma proves that always, no matter which graph was sampled from the distribution {$\ExtGH$}, the value of a fractional solution is at most $O(n^2)$.

\begin{lemma}\label{lem:fracsol}
With a probability of $1$ over the distribution defined by {$\ExtGH$}, there is a feasible fractional solution to $(MET)$ and instance $\mathcal{I}(G_X,\ell,L)$ of value at most $O(n^2)$. 
\end{lemma}
\begin{proof}
We define the following solution to $(MET)$, no matter which graph was sampled from the distribution {$\ExtGH$}: set $\delta(u,v)$ to be the shortest path metric for the graph $\mathcal{G}$ where the length of an edge $e=(u,v)\in E$ equals $\ell(e)$ and the length of an edge $e=(v,v_T)\in E$ equals $L$.

We start by proving that the above solution is feasible.
First, note that $(V,\delta)$ is a semi-metric space.
Second, one can verify that $\delta(u_T,v_T)=D(u_T,v_T)$, for every $u_T,v_T\in T$.
Thus, the above solution is feasible for $(MET)$.

Now, let us bound the cost of the above feasible solution.
Recalling the definition of weights $w$ of edges of $\mathcal{G}$, one can see that for every edge $e\in E$ its contribution $w_e\cdot \delta(e)$ to the objective of $(MET)$ equals $1$.
Hence, since the degree of $G$ and $H$ is a constant $d$, the cost of the above solution is at most the number of edges in $\mathcal{G}$, which is upper bounded by $O(n^2) $.
\end{proof}

\section{Cycle-Homeomorphism, Split, and the Analysis of Integral Solutions}\label{sec:IntegralSol}
In this section we introduce the key definitions we require for analyzing integral solutions.
Our ultimate goal is the following lemma, which proves that with a non-zero probability, over the distribution defined by {$\ExtGH$}, any solution to {\EXT} for the instance $\mathcal{I}(G_X,\ell,L)$ has a large enough value.
One can easily observe that our main result, Theorem \ref{thrm:Main}, follows from lemmas \ref{lem:fracsol} and \ref{lem:integsol}.

\begin{lemma}\label{lem:integsol}
With a non-zero probability over the distribution defined by {$\ExtGH$} the value of any integral solution to {\EXT} for the instance $\mathcal{I}(G_X,\ell,L)$ has value of at least $\Omega(n^2\log^{\nicefrac[]{2}{3}}n)$.
\end{lemma}

Recall that an integral solution $f:V\rightarrow T$ assigns to every vertex a terminal in $T$.
To simplify the presentation, we abuse notations and refer to a terminal $v_T\in T$ as belonging to the same cloud as $v$ (recall that the only edge in $\mathcal{G}$ that touches $v_T$ is $(v,v_T)\in E$).
Thus, for ease of presentation we use $v$ instead of $v_T$ when possible, and denote by $D(u,v)$ the distance in $D$ between $u_T$ and $v_T$, {\em i.e.}, $D(u_T,v_T)$.

In this section we prove the following crucial key insight, which for simplicity of presentation we currently present in its qualitative form: every cheap integral solution $f$ assigns most of the vertices of a cloud to the same terminal and this holds for most of the clouds.
We denote this terminal, if it exists, as the {\em representative} of the cloud, with respect to the given $f$.
Moreover, we prove that the map from the clouds to their representatives keeps some structure of $G$ that is required for our analysis. This is captured by the following definition, which is central to our analysis.

\begin{definition}\label{Def:CycleHomomorphism}
Let $G = (V_G, E_G)$ be a graph and $G'=(V_{G'},E_{G'})$ a subgraph of $G$ and $\tilde{f}: V_{G'} \to V_G$ a mapping.
Moreover, $\tilde{f}$ is associated with a mapping $P_{\tilde{f}}: E_{G'}\to E_G^*$, where for every $(u,v)\in E_{G'}$ $P_{\tilde{f}}(e)$ is a path in $G$ between $\tilde{f}(u)$ and $\tilde{f}(v)$.
We say that $(\tilde{f},P_{\tilde{f}})$ is a {\em cycle-homeomorphism of $G$ with respect to $G'$} if for every simple cycle $C\subseteq E_{G'}$:
$$\left\{ e\in E_G:e\text{ appears an odd number of times in } \{P_{\tilde{f}}(e')\}_{e'\in C}\right\}=C .$$
\end{definition}
The above definition leads to the definition of a {\em split} of a graph, which plays an essential role in our analysis.
\begin{definition}\label{def:Split}
Let $G_X$ and $\ell$ be a weighted graph in the support of {$\ExtGH$}, and let $(T,D)$ be the metric space defined in Section \ref{sec:GapInstance}.
A map $\bar{f}: V_{G'}  \to X$ is an {\em $(\alpha, \varepsilon)$-split} if there exists a subgraph $G'=(V_{G'},E_{G'})$ of $G$, where $|E_{G'}|\geq (1-\varepsilon)|E_G|$, such that the following two conditions hold:
\begin{enumerate}
\item Let $\pi:X\rightarrow V_G$ be the projection of the vertices of $G_X$ to the clouds: $\pi ((g,h))=g$ for every $(g,h)\in X$.
Given $\pi\circ \bar{f}$, we define $P_{\pi\circ \bar{f}}$ to be the image under $\pi$ of the shortest paths in $G_X$: for every $(g_1,g_2)\in E_{G'}$, $P_{\pi\circ \bar{f}}((g_1,g_2))$ equals the image under $\pi$ of the shortest path with respect to $\ell$ in $G_X$ between $\bar{f}(g_1)$ and $\bar{f}(g_2)$.
Then $(\pi\circ \bar{f},P_{\pi\circ \bar{f}})$ is a cycle-homeomorphism of $G$ with respect to $G'$.
\item For all edges $(g_1,g_2)\in E_{G'}$ of $G'$: $D_X(\bar{f}(g_1),\bar{f}(g_2))< \alpha$.
\item For all $g\in V_{G'}$: the shortest path in $G$ between $g$ and $\pi \circ \bar{f}(g)$ has at most $\varepsilon \log{n}$ edges.
\end{enumerate}

\end{definition}

One important thing to note is that in the first condition of the above definition, $P_{\pi\circ \bar{f}}$ is defined for every edge of $G'$, however the shortest path $P_{\pi\circ \bar{f}}$ assigned to an edge in $G'$ is the image under $\pi$ of a path in the entire of $G_X$ which is in the support of {$\ExtGH$} (and not {$\ExtGprimeH$}).

Intuitively, we prove that for every cheap solution $f$ it is possible to throw away a fraction of $O(\varepsilon)$ of the clouds such that all remaining clouds have a representative with respect to $f$.
Additionally, the map from the remaining clouds to their representatives is a cycle-homeomorphism.
This is summarized in the following theorem.

\begin{theorem}\label{thrm:Split}
There exists a small absolute constant $c$ such that for every small enough absolute constant $\varepsilon >0$, every large enough $n$ (that might depend on $\varepsilon$), and every weighted $G_X$ and $\ell$ in the support of {$\ExtGH$} and every integral solution $f:V\rightarrow T$ for $\mathcal{I}(G_X,\ell,L)$ whose value is at most $c\varepsilon^2n^2\log^{\nicefrac[]{2}{3}}n$, there exists a map $\bar{f}:V_{G'}\rightarrow X$ such that $\bar{f}$ is an $(\varepsilon \log ^{\nicefrac[]{4}{3}} n,\varepsilon)$-split.
\end{theorem}

It is important to note that the above theorem holds for every possible realization of the distribution of {$\ExtGH$}. 
Theorem \ref{thrm:Split} stands in contradiction to the following theorem, whose proof Section \ref{sec:NoSplit} is dedicated to.

\begin{theorem}\label{thrm:NoSplit}
For every small enough absolute constant $\varepsilon >0$, and for every large enough $n$ (that might depend on $\varepsilon$),
with a non-zero probability there is no $\bar{f}:V_{G'}\rightarrow X$
such that $\bar{f}$ is an $(\varepsilon \log ^{\nicefrac[]{4}{3}} {n},\varepsilon)$-split.
\end{theorem}
In the above two theorems, the reader should recall that $G$ and $H$, along with $\ell_G$ and $\ell_H$, need to satisfy all the conditions as described in Section \ref{sec:GapInstance}.
One can easily note that Theorems \ref{thrm:Split} and \ref{thrm:NoSplit} imply Lemma \ref{lem:integsol}.
\begin{proof}[Proof of Lemma \ref{lem:integsol}]
Follows immediately from Theorems \ref{thrm:Split} and \ref{thrm:NoSplit}.
\end{proof}

In the remainder of this section, we prove Theorem \ref{thrm:Split}.
For simplicity of presentation, from this point onward we fix a weighted graph $G_X$ sampled from {$\ExtGH$} (along with its weight function $\ell$).
Let $f:V\rightarrow T$ be an integral solution to $\mathcal{I}(G_X,\ell,L)$ whose cost is at most $c \varepsilon^2 n^2\log^{\nicefrac[]{2}{3}}n$. 
The following lemma states that for such an integral solution $f$, most of the clouds have a representative, {\em i.e.}, at least a fraction of $(1-O(\varepsilon))$ of the vertices of the cloud are assigned by $f$ to the same terminal.
Moreover, almost all of these representatives are not far from the cloud they represent.
We note that in the following lemma we use the notation of $ \bar{f}(g)$ to denote the representative of a cloud $g$ and by $\pi \circ \bar{f}(g)$ the cloud this representative resides in.
We intentionally chose this notation (as in Definition \ref{def:Split}) since in the proof of Theorem \ref{thrm:Split} it is shown that this $\bar{f}$ is indeed a split.

\begin{lemma}\label{lemma:CldsHasReps}
Let 
\begin{align*}
S\triangleq \{ g\in V_G:&\exists v_T\in T \text{ s.t. }\\
& \left|\{ (g,h):f((g,h))=v_T\}\right| \geq \left(1-O(\varepsilon)\right)\cdot |V_H|\}
\end{align*}
be the collection of clouds that have a representative, and for each $g\in S$ we denote its representative by $\bar{f}(g)$ 
and by $\pi \circ \bar{f}(g)$ 
the cloud this representative resides in.
Then the set
\begin{align*}
 \big\{ g\in S: & {\text{ the shortest path in $G$ between $g$ and $ \pi \circ \bar{f}(g)$  }}\\
 & {\text{ has at most $\varepsilon \log n$ edges }}\big\}
\end{align*}
is of size at least $ (1-O(c\varepsilon))\cdot |V_G|$.
\end{lemma}
\begin{proof}
Our proof is in two stages: $(1)$ we prove that $|S|\geq (1-O(c\varepsilon))|V_G|$; and $(2)$ we use the latter lower bound on the size of $S$ to conclude the proof.

Let us focus on the first stage, proving a lower bound on $|S|$.
Fix a cloud $g$, and assume that $g\notin S$, {\em i.e.}, there is no terminal $v_T$ such that at least a fraction of $(1-O(\varepsilon))$ of the vertices in the cloud are assigned to $v_T$.
Since the cloud $g$ is a copy of $H$ and $H$ is an expander, we can assume that there are at least $\Omega (\varepsilon |V_H|)$ intra-cloud edges inside $g$ whose two endpoints are assigned to different terminals.
Recalling the definition of the metric $D$ over the terminals, we know that for any $v_T\neq v'_T$: $ D(v_T,v'_T)\geq 2L=\Omega (\log{n})$.
The weight $w$ of an intra-cloud edge equals $\log^{-\nicefrac[]{1}{3}}n$, thus we can conclude that such a cloud $g$ contributes to the value of the solution $f$ at least $\Omega (\varepsilon n \log ^{\nicefrac[]{2}{3}}n)$.
If there are too many such clouds $g$, {\em i.e.}, at least $\Omega (c\varepsilon n)$, this causes the value of the solution $f$ to exceed $c\varepsilon ^2 n^2 \log ^{\nicefrac[]{2}{3}}n$, a contradiction.
Thus, we can can assume that there are at most $O(c\varepsilon n)$ such clouds.
This proves that $|S|\geq (1-O(c\varepsilon))n$.
This concludes the first stage of the proof.

Let us now focus on the second stage of the proof.
Fix a cloud $g\in S$, and assume that the number of edges in the shortest path in $G$ between $g$ and $\pi \circ \bar{f}(g)$
is more than $\varepsilon \log{n}$.
For every vertex $v$ in the cloud $g$ such that $ f(v)=\bar{f}(g)$
we consider the edge $(v,v_T)\in E$.
The distance with respect to $D$ between the terminals assigned to the two endpoints of this edge, {\em i.e.}, $v$ and $v_T$, is at least $\varepsilon \log ^{\nicefrac[]{5}{3}}n$.
The weight $w$ of such an edge equals $1/\log{n}$, thus its contribution to the value of $f$ is at least $\varepsilon \log^{\nicefrac[]{2}{3}}n$.
Since there are at least $(1-O(\varepsilon))|V_H|$ vertices in the cloud $g$ that $f$ assigns to  $\bar{f}(g)$,
we can conclude that this cloud contributes to the value of $f$ at least $(1-O(\varepsilon))\varepsilon n \log ^{\nicefrac[]{2}{3}}n$.

Assume to the contrary that there are at least $\Omega (c\varepsilon n)$ clouds $g$ such that the number of edges in the shortest path in $G$ between $g$ and $ \pi \circ \bar{f}(g)$
is more than $\varepsilon \log{n}$.
This implies that the cost of the solution $f$ exceeds $c\varepsilon ^2 n^2 \log^{\nicefrac[]{2}{3}}n$, which is a contradiction.
Hence, we can conclude that there are only at most $O(c\varepsilon n)$ such clouds.
\end{proof}

Equipped with Lemma \ref{lemma:CldsHasReps}, we are now ready to prove Theorem \ref{thrm:Split}.

\begin{proof}[Proof of Theorem \ref{thrm:Split}]
First let us define $G'$ to be the subgraph of $G$ that contains all clouds that have a close representative, {\em i.e.}, $V_{G'}$ contains all clouds $g\in S$, where $S$ and $\bar{f}$ (as well as $\pi \circ \bar{f}$) are as in Lemma \ref{lemma:CldsHasReps}, such that the shortest path in $G$ between $g$ and $ \pi \circ \bar{f}(g)$ 
contains at most $\varepsilon \log{n}$ edges.
Additionally, we have in $G'$ all edges $(g_1,g_2)\in E_G$ such that $g_1,g_2\in V_{G'}$ and $D_X(\bar{f}(g_1),\bar{f}(g_2)) < \varepsilon \log^{\nicefrac[]{4}{3}}n$.

First, let us prove that $G'$ contains enough edges as required by Definition \ref{def:Split}, 
{\em i.e.}, $|E_{G'}|\geq (1-\varepsilon)|E_G|$.
Recall that Lemma \ref{lemma:CldsHasReps} implies that $|V_{G'}|\geq (1-O(c\varepsilon))n$.
Let us denote by $B$ the number of edges $(g_1,g_2)\in E_G$ such that $g_1,g_2\in V_{G'}$, however $ D(\bar{f}(g_1),\bar{f}(g_2)) \geq D_X (\bar{f}(g_1),\bar{f}(g_2))\geq \varepsilon \log^{\nicefrac[]{4}{3}}n$.
We prove that there are at most $(1+O(\varepsilon))c\varepsilon n$ such edges.
Every such edge contributes to the cost of the solution $f$ at least:
 $(\varepsilon \log ^{\nicefrac[]{4}{3}}n) \cdot ( \log ^{-\nicefrac[]{2}{3}}n)\cdot \left( 1-O(\varepsilon)\right) n $, where the $\varepsilon \log ^{\nicefrac[]{4}{3}}n$ term is the lower bound on the distance $D$ between the terminals assigned to the endpoints, the $\log ^{-\nicefrac[]{2}{3}}n$ term is the weight $w$, and the $\left( 1-O(\varepsilon)\right) n$ term is the number of matching edges between the clouds $g_1$ and $g_2$ that both their endpoints are assigned to the corresponding representative.
Thus, since the total cost of $f$ is at most $c\varepsilon ^2 n^2 \log ^{\nicefrac[]{2}{3}}n$, one can deduce that $B\leq (1+O(\varepsilon))c\varepsilon n$.
Thus, for a small enough constant $c$ we have that $|E_{G'}|\geq (1-\varepsilon) |E_G|$.

Second, let us focus on the three conditions of Definition \ref{def:Split},
and prove that $\bar{f}$ is indeed an $(\varepsilon \log^{\nicefrac[]{4}{3}}n,\varepsilon)$-split.

Let us start with the second condition of Definition \ref{def:Split}, which states that for every $(g_1,g_2)\in E_{G'}$ it holds that $D_X(\bar{f}(g_1),\bar{f}(g_2)) < \alpha$ where $\alpha = \varepsilon \log^{\nicefrac[]{4}{3}}n$.
Note that this condition holds trivially by definition of $G'$.
Moreover, one can easily see that the third condition of Definition \ref{def:Split} also trivially holds.

Let us consider the first condition of Definition \ref{def:Split}.
We prove that $\bar{f}$ satisfies that $(\pi\circ \bar{f},P_{\pi\circ \bar{f}})$ is a cycle-homeomorphism of $G$ with respect to $G'$.
Fix a simple cycle $C\subseteq E_{G'}$ in $G'$, and consider an arbitrary edge $(g_1,g_2)\in C$.
We examine now two paths between $\pi\circ \bar{f}(g_1)$ and $\pi\circ \bar{f}(g_2)$ in $G$.
The first is defined as the concatenation of the shortest path in $G$ between $\pi\circ \bar{f}(g_1)$ and $g_1$, with the edge $(g_1,g_2)$, and with the shortest path in $G$ between $g_2$ and $\pi \circ \bar{f}(g_2)$.
Note that this path contains at most $1+2\varepsilon \log{n}$ edges of $G$.
The second path is defined by taking the shortest path in $G_X$ between $\bar{f}(g_1)$ and $ \bar{f}(g_2)$ and projecting it to $G$ via $\pi$.
Recall that this is exactly $P_{\pi\circ \bar{f}}((g_1,g_2))$.
Note that this path contains at most $\varepsilon \log ^{\nicefrac[]{2}{3}}n$ edges of $G$ (recall that $D_X(\bar{f}(g_1),\bar{f}(g_2))\leq \varepsilon \log^{\nicefrac[]{4}{3}}n$ and the length of every edge in $G$ equals $\log^{\nicefrac[]{2}{3}}n$).

Equipped with the above two paths in $G$ between $\pi\circ \bar{f}(g_1)$ and $\pi\circ \bar{f}(g_2)$, we examine their concatenation.
This results in a cycle that contains at most $1+2\varepsilon \log{n}+\varepsilon \log^{\nicefrac[]{2}{3}}n$ edges in $G$.
Hence, for a small enough $\varepsilon$ this concatenated cycle contains less edges than the girth of $G$.
This implies that this concatenated cycle is not simple and every edge in it appears an even number of times (one can view these paths in Figure \ref{fig:CycleHomeo}).

\begin{figure}[t]
    \centering
    \includegraphics[trim={6cm 0 6cm 0},clip,width=0.9\linewidth]{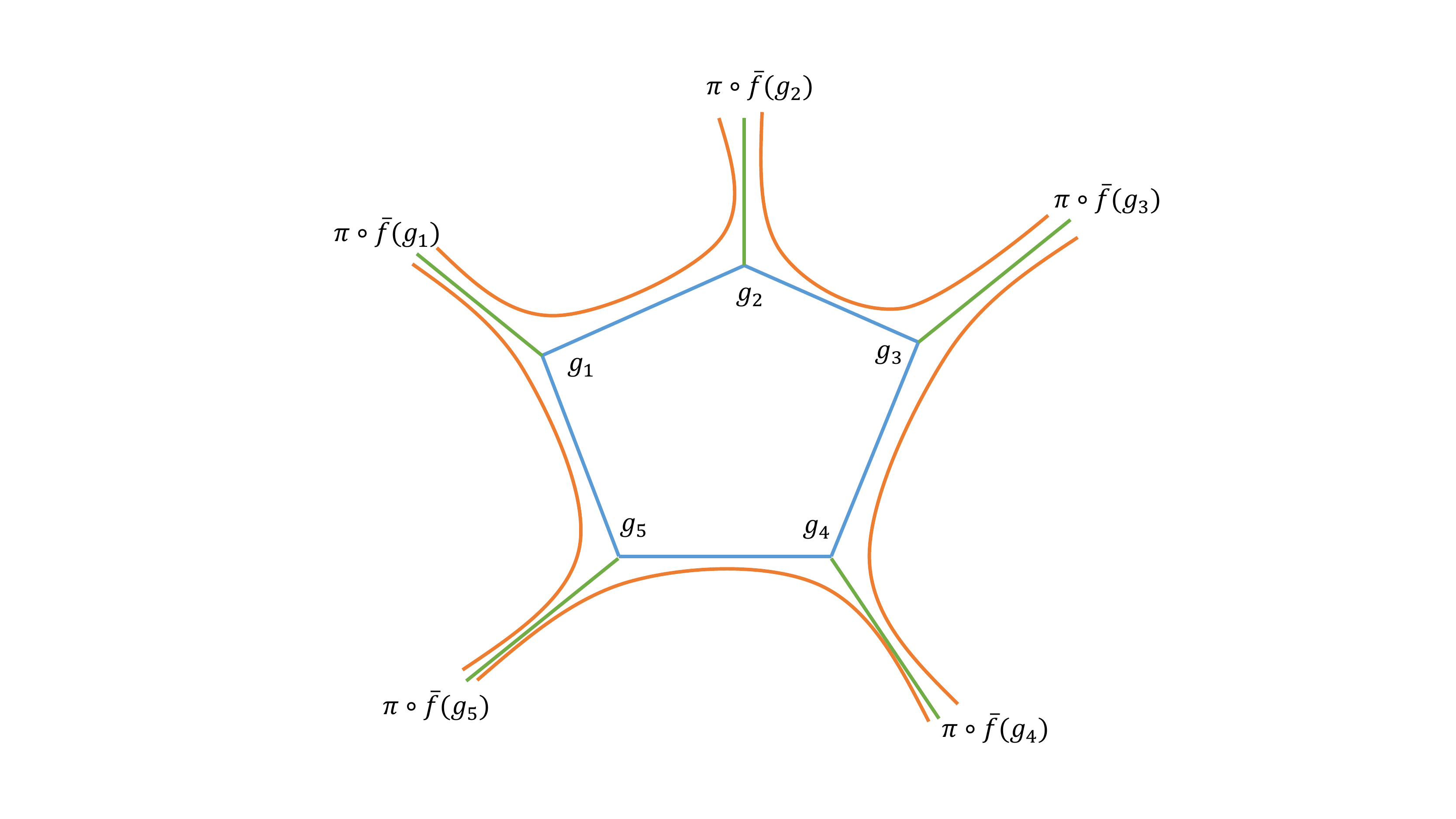}
    \caption{$C$ colored blue, short paths between $\pi\circ \bar{f}(g_i)$ and $g_i$ colored green, and $P_{\pi\circ \bar{f}}$ colored orange}
    \label{fig:CycleHomeo}
\end{figure}

Adding up over all edges of $C$, the concatenated cycle of each such edge, we obtain a collection of edges (counting multiplicities) that: $(1)$ every edge in the collection appears an even number of times; $(2)$ the collection is comprised of adding up the cycle $C$, $\cup _{e\in C}P_{\pi\circ \bar{f}}(e) $ (with multiplicities), and twice each of the shortest path in $G$ between every vertex $g\in C$ and $ \pi \circ \bar{f}(g)$.
Using $(1)$ and $(2)$ above we can conclude that all edges appearing an odd number of times in $ \cup _{e\in C}P_{\pi \circ \bar{f}}(e)$ are exactly all edges of $C$.
This proves that $(\pi\circ \bar{f},P_{\pi \circ \bar{f}})$ is a cycle-homeomorphism of $G$ with respect to $G'$.
This concludes the proof as $\bar{f}$ is an $(\varepsilon \log^{\nicefrac[]{4}{3}}n,\varepsilon)$-split.
\end{proof}

\section{Certificates and the Proof of Theorem \ref{thrm:NoSplit}}\label{sec:NoSplit}
In this section we introduce the notion of a certificate. which is used to prove Theorem \ref{thrm:NoSplit}.
In order to prove Theorem \ref{thrm:NoSplit}, we are required to upper bound the probability that there are short paths between our representatives (recall the second requirement in Definition \ref{def:Split} of a split).
The above is achieved by carefully ``counting'' all possible paths, to which end the notion of a certificate is useful.

More specifically, given a sampled $G_X$ and a $(\varepsilon \log^{\nicefrac[]{4}{3}}n,\varepsilon)$-split $\bar{f}$ (recall that $\bar{f}$ is also associated with an appropriate  subgraph $G'$ of $G$), we construct a combinatorial structure which we call a certificate.
Intuitively, this certificate contains only {\em partial} information given $G_X$ and the split $\bar{f}$, which is enough to reconstruct the shortest path in $G_X$ between the representatives of neighboring clouds.
Next, we prove two important facts.
First, given $G$ and $H$ there is an upper bound on the number of possible certificates.
Second, given a certificate the probability (over sampling from $\ExtGH$) of obtaining the given certificate from $G_X$ and any $(\varepsilon \log^{\nicefrac[]{4}{3}}n,\varepsilon)$-split $\bar{f}$ is sufficiently low.
To conclude our argument we simply employ the union bound and prove that with a positive probability we sampled a $G_X$ that cannot produce any certificate, thus this $G_X$ does not have any $(\varepsilon \log^{\nicefrac[]{4}{3}}n,\varepsilon)$-split $\bar{f}$.

The reader should note that in this section, since we are focusing on Theorem \ref{thrm:NoSplit}, the instance $\mathcal{I}(G_X,\ell,L)$ does not play a role whereas $G_X$ and $\ell$ do.

\subsection{Building Certificates}

We start with how a certificate is constructed given $G_X$ and a split $\bar{f}$.
To simplify the definition of a certificate, we recall that both $G$ and $H$ are Cayley graphs of some group, and that every vertex in $V_G$ and $V_H$ is uniquely identified with an element of the corresponding group.
The property (which is stated for $H$ but also applies to $G$) we require in our analysis is the following: given a path $P$ in $H$ that starts at $u$ and ends in $v$, we can determine:
$(1)$ if there is an edge connecting $u$ and $v$ in $H$; and $(2)$ what are the two labels of this edge, assuming it exists.
Moreover, it is useful to note that a path $P$ from $u$ to $v$ in $H$ corresponds to $x=u^{-1}\cdot v$ in the group, or equivalently, that multiplying the labels of the edges of $P$ in direction from $u$ to $v$ equals to $x$.
Thus, for example, if $P$ is a cycle then $x$ is the identity (recall Definition \ref{Def:CycleHomomorphism} of cycle-homeomorphism).

We denote the labeling of an edge $e\in E_H$ and one of its two end points $u\in V_H$ by $\LH(u,e)$ and the labeling of an edge $e\in E_G$ and one of its two end points $u\in V_G$ as $\LG(u,e)$.
One can note that this induces a labeling of $G_X$, where every vertex $(g,h)\in X$ is associated with the ordered pair of group elements $(g,h)$ and every edge in $E_X$ is associated with either an element from the group $G$ is its Cayley graph, or the group $H$ is its Cayley graph (the former is for inter-cloud edges whereas the latter is for intra-cloud edges).
Hence, vertex $(g,h)\in X$ has exactly $2d$ different labels on the $2d$ edges touching it in $G_X$, $d$ due to $H$ (intra-cloud edges) and $d$ due to $G$ (inter-cloud edges).
We denote this labeling by $\LX((g,h),e)$.
For simplicity of presentation, for an edge $e\in E_X$ we denote by $\LX(e)$ the label $\LX((g,h),e)$, where $(g,h)$ is one of the two end vertices of $e$, where it is clear from the context which of the two end vertices is chosen.
Alternatively, one can associate a direction for every edge $e\in E_X$ that is clear from the context and dictates which of the two end vertices of $e$ is chosen for $\LX(e)$.

\vspace{5pt}
\noindent {\bf{Formal Transformations.~~~~~}}
Given a sampled $G_X$ and a $(\varepsilon \log^{\nicefrac[]{4}{3}}n,\varepsilon)$-split $\bar{f}$, we consider the collection of all shortest paths in $G_X$ between representatives of neighboring clouds (as given by $G'$).
The first step in constructing a certificate is, intuitively, to strip information from these paths: the identities of vertices inside each cloud are removed while the labels on the intra-cloud edges remain (thus {\em absolute} information inside the cloud is erased but {\em relative} information inside the cloud remains).
This is achieved by the notion of a formal transformation and is captured by Definition \ref{def:FormTrans} and Algorithm \ref{alg:FormalTrans}.
In what follows, the reader should keep in mind that the collection of paths $\{ P_{\ell}\} _{\ell=1}^M$ that will be plugged into Algorithm \ref{alg:FormalTrans} is the collection of shortest paths in $G_X$ between the representatives of neighboring clouds whose existence is ensured by the split $\bar{f}$ and $G'$ (recall Definition \ref{def:Split}).
In what follows we assume every path is directed in an arbitrary direction which is fixed, thus $\LX(e)$ is well defined in the following definition of a formal transformation.

\begin{definition}\label{def:FormTrans}
The {\em formal transformation} of a collection of paths $\{ P_{\ell}\} _{\ell = 1}^M$ in a sampled graph $G_X$ from $ \ExtGH$ is a collection $\{ Q_{\ell}\} _{\ell = 1}^M$ of paths where:
\begin{enumerate}
    \item every vertex $(g,h)$ in $P_{\ell}$ is given in $Q_{\ell}$ only by an index $\text{ind}((g,h))=(g,i)$ where $i$ is an index in $\{ 1,\ldots,|V_H|\}$. 
    \item if a vertex $(g,h)$ appears in paths $P_{\ell}$ and $P_{\ell'}$ then both in $Q_{\ell}$ and $Q_{\ell'}$ the vertex $(g,h)$ is given by the same index $\text{ind}((g,h))=(g,i)$ for some $i$.
    \item every edge $e=((g,h),(g',h'))$ in a path $P_{\ell}$ is given in $Q_{\ell}$ also by $\LX(e)$. 
\end{enumerate}
\end{definition}

Let us now focus on Algorithm \ref{alg:FormalTrans}.
Its output is the {\em formal transformation} $\{ Q_{\ell}\} _{\ell =1}^M$ of the input $\{ P_{\ell}\} _{\ell=1}^M$, where the indices $i$ are given sequentially to each vertex according to the order they are exposed.
In Algorithm \ref{alg:FormalTrans}, $V_{\text{exposed}}$ is the collection of vertices seen so far, $i$ is a running index that produces the sequential numbering for each cloud, and $\text{ind}:X \rightarrow V_G\times \mathbb{N}$ is the indexing $(g,i)$ the algorithm produces in the formal transformation.
Observe that if one knows a sampled $G_X$ from $\ExtGH$, the output $\{ Q_{\ell}\} _{\ell =1}^M$ of Algorithm \ref{alg:FormalTrans} (for an unknown input $\{ P_{\ell}\}_{\ell = 1}^M$), and the {\em true} identity of at least one of the endpoints of each path in $\{ Q_{\ell}\} _{\ell = 1}^M$ (by true identity we mean that if $(g,h)$ is the start or end vertex of a path $P_{\ell}$ then $h$ is also known in addition to the information given by the formal transformation of $P_{\ell}$), one can reconstruct $\{ P_{\ell}\}_{\ell = 1}^M$.
The following lemma summarizes the guarantee of Algorithm \ref{alg:FormalTrans} along with an upper bound on the number of vertices that can appear in every cloud $g\in V_G$.

\SetKwInput{KwInput}{Input}                
\SetKwInput{KwOutput}{Output}              

\begin{algorithm}[ht]
\caption{FormalTransAlg}
\DontPrintSemicolon
\label{alg:FormalTrans}
\KwInput{$\{ P_{\ell} \}_{\ell =1}^M$}
\KwOutput{formal transformation $\{ Q_{\ell}\} _{\ell =1}^M$ of $\{ P_{\ell} \}_{\ell =1}^M$}
 {\bf{initialization:}}  $i_g\leftarrow 1$ $\forall g\in V_G$ and $V_{\text{exposed}} \leftarrow \emptyset$\;
    \For{$\ell=1$ to $M$}{
    initialize $Q_{\ell}$ to be an empty path\;
    let $(g,h)$ be the first vertex in $P_{\ell}$\;
    \uIf{$(g,h)\notin V_{\text{exposed}}$ }{
            $V_{\text{exposed}}\leftarrow V_{\text{exposed}} \cup \{ (g,h)\}$\;
            $\text{ind}((g,h))\leftarrow (g,i_g)$\;
            $i_g\leftarrow i_g+1$\;
        }
     \For{$e = ((g,h),(g',h'))$ the next edge of $P_{\ell}$}{
        \uIf{$(g',h')\notin V_{\text{exposed}}$ }{
            $V_{\text{exposed}}\leftarrow V_{\text{exposed}} \cup \{ (g',h')\}$\;
            $\text{ind}((g',h'))\leftarrow (g',i_{g'})$\;
            $i_{g'}\leftarrow i_{g'}+1$\;
        }
        in $Q_{\ell}$ add the edge $(\text{ind}((g,h)),\text{ind}((g',h')))$ and label it $\LX(e)$\; 
     }
 }
\Return  $Q_1,\ldots,Q_M$\;
\end{algorithm}

\begin{lemma}\label{lem:versInCloud}
Let $G_X$ be a graph in the support of $ \ExtGH$, and let $\bar{f}$ be a $(\varepsilon \log ^{\nicefrac[]{4}{3}}n,\varepsilon)$-split and let $G'$ be the subgraph of $G$ associated with $\bar{f}$.
Let $\{ P_{\ell}\} _{\ell=1}^M$ be the collection of shortest paths in $G_X$ between $\bar{f}(g_1) $ and $ \bar{f}(g_2)$, for every $(g_1,g_2)\in E_{G'}$.
Then Algorithm \ref{alg:FormalTrans} when applied to $\{ P_{\ell}\} _{\ell =1}^M$ outputs the formal transformation of $ \{ P_{\ell}\} _{\ell = 1}^M$.
Moreover, every cloud $g\in V_G$ has at most $n^{O(\varepsilon)}$ vertices in $\{ Q_{\ell}\} _{\ell = 1}^M$.

\end{lemma}
\begin{proof}
First, it is clear that Algorithm \ref{alg:FormalTrans} produces a formal transformation, according to Definition \ref{def:FormTrans}, of $\{ P_{\ell}\} _{\ell = 1}^M$.
Second, given an edge $(g_1,g_2)\in 
E_{G'}$, recall that by Definition \ref{def:Split} $ D_X(\bar{f}(g_1),\bar{f}(g_2)) < \varepsilon \log^{\nicefrac[]{4}{3}}n$.
Since the shortest length of an edge in $G_X$ equals $\log ^ {\nicefrac[]{1}{3}} n $, one can conclude that the shortest path in $G_X$ between $\bar{f}(g_1)$ and $\bar{f}(g_2)$ has at most $\varepsilon \log n$ edges.
Given a cloud $g\in V_G$, we now aim to upper bound the number of paths in $\{ P_{\ell}\}_{\ell =1}^M$ that contain at least one vertex from $g$.
Every such path corresponds to an edge $(g_1,g_2)\in E_{G'}$ satisfying: $\pi \circ \bar{f}(g_1)$ or $\pi \circ \bar{f}(g_2)$ is within $\varepsilon \log{n}$ edges away from $g$ in $G$.
Moreover, $\pi \circ \bar{f}(g_1)$ (or alternatively $ \pi \circ \bar{f}(g_2)$) is within at most $\varepsilon \log{n}$ edges away from $g_1$ in $G$ (or alternatively $g_2$ in $G$), this follows from the third condition in Definition \ref{def:Split}.
Thus, the total number of $(g_1,g_2)\in E_{G'}$ edges whose path between $\bar{f}(g_1)$ and $\bar{f}(g_2)$ passes through the cloud $g$ is upper bounded by $d\cdot d^{2\varepsilon \log{n}} = n^{O(\varepsilon)}$.
To conclude the proof, we need to recall again that each path in $\{ P_{\ell}\} _{\ell = 1}^M$ contains at most $\varepsilon \log{n}$ edges of $G_X$ and thus the number of vertices in each cloud that belong to $\{Q_{\ell} \} _{\ell = 1}^M$ is at most $\varepsilon \log{n} \cdot n^{O(\varepsilon)} = n^{O(\varepsilon)}$.
\end{proof}

\vspace{5pt}
\noindent {\bf{Inner Connected Components.~~~~~}}
In this section we define when a connected component of $\{ Q_{\ell}\} _{\ell =1}^M$, when restricted to a cloud $g$, is significant to our analysis.
We recall that given $\{Q_{\ell} \} _{\ell = 1}^M$, every vertex in $\{Q_{\ell} \} _{\ell = 1}^M$ is given by $\text{ind}(g,h)=(g,i)$ (for some $i$) and not by $(g,h)$, {\em i.e.}, the cloud $g$ it belongs to and a serial number $i$ given to it by Algorithm \ref{alg:FormalTrans}.
Fix a cloud $g$ and examine the connected components of the graph whose vertices are all the vertices of $\{ Q_{\ell}\} _{\ell =1}^M$ which belong to cloud $g$ and all the edges appearing in $\{ Q_{\ell}\} _{\ell =1}^M$ whose two endpoints are in cloud $g$.
We define the {\em degree} of a connected component as the number of distinct inter-cloud edges that appear in $\{ Q_{\ell}\} _{\ell =1}^M$ and whose one of their end vertices belongs to the connected component.
We say that a connected component is an {\em inner connected component} if its degree is at least three or if it contains a representative of a cloud, {\em i.e.}, it contains a starting vertex or end vertex of a path $P_{\ell}$.

We note that a connected component which is not inner must be a path.
The reason for that is that every vertex in the connected component belongs to at least one path $Q_{\ell}$.
Since there are no representatives in the component, it must be the case that for each such $Q_{\ell}$ and for each maximal subpath of $Q_{\ell}$ that is contained as a whole in the component, there are two distinct inter-cloud edges in $Q_{\ell}$ (as $Q_{\ell}$ is simple), one immediately preceding it and the other immediately following it.
This implies that the degree of the component is at least two.
Since the degree of the component is at most two (recall it is not inner), we can conclude that its degree is exactly two.
Since $H$ has girth greater than $2\varepsilon \log{n}$ (for every small enough constant $\varepsilon >0$), and recalling that every $Q_{\ell}$ contains at most $\varepsilon \log{n}$ edges, one can prove that the connected component is exactly the unique shortest path (inside the cloud) between the endpoints that are in the cloud of the two inter-cloud edges touching the component.

Let us define the {\em inner connected components graph} $R=(V_R,E_R)$ as follows.
The vertices $V_R$ of this graph are the inner connected components as defined above.
For every subpath of $\{ Q_{\ell}\} _{\ell =1}^M$ starting from an inter-cloud edge leaving one inner connected component $C_1$ and ending in an inter-cloud edge entering an inner connected component $C_2$ (and not passing through another inner connected component in between) we add an edge between $C_1$ and $C_2$ (which might be a self loop in case $ C_1=C_2$). Note that the degree of an inner connected component in the graph coincides with our previous definition of its degree.
We call an edge in $\{ Q_{\ell}\} _{\ell =1}^M$ {\em surprising} if it is the first or the last edge in a subpath corresponding to an edge in $E_R$.
For the remainder of the paper we denote by $s_{\text{tot}}$ the sum of the degrees of the inner connected components.
Note that $s_{\text{tot}}=2|E_R|$.

We say that a vertex in $\{ Q_{\ell}\} _{\ell = 1}^M$ (which is given only by the cloud $g$ it belongs to and the serial number $i_g$ given to it by Algorithm \ref{alg:FormalTrans}) that belongs to an inner connected component is {\em distinguished} if it satisfies one of the following two conditions: $(1)$ the vertex is a representative, {\em i.e.} it is the first or last vertex in some path in $\{ P_{\ell}\}_{\ell=1}^M$; and $(2)$ the vertex touches a surprising edge.
We note that every inner connected component has at least one distinguished vertex that belongs to it.

Let us now define what a {\em representation} of an inner connected component is.
Given an inner connected component, its representation is comprised of two things: $(1)$ the collection of distinguished vertices that belong to the inner connected component and their $\text{ind}(g,h)$ as given in $\{ Q_{\ell}\}_{\ell=1}^M$ (recall that only the cloud $g$ and a running index is given by $\text{ind}(g,h)$ and not the true identity of the vertex in the group $H$ is its Cayley graph);
and $(2)$ for every ordered pair of distinguished vertices in the inner connected component the difference between them in the group whose Cayley graph is $H$, {\em i.e.}, the sum of $ \LX(e)$ over edges $e$ in the appropriate subpath of $Q_{\ell}$.
The following lemma provides an upper bound on the number of representations, and it will enable us to eventaully upper bound the number of certificates.

\begin{lemma}\label{Lemma:CompRepCount}
Denote by $N(s,r)$ the number of distinct representations of an inner connected component that: $(1)$ has $r$ representatives; $(2)$ has $s$ surprising edges touching it in $R$; amd $(3)$ can be obtained with a positive probability over the sampling of $G_X$ from $\ExtGH$.
Then $N(s,r)\leq n^{O(\varepsilon (s+r))}$.
\end{lemma}
\begin{proof}
First, we note that every vertex in the inner connected component is at most $\varepsilon \log{n}$ edges away (using only edges in $E_H$, {\em i.e.}, intra cloud edges) from a distinguished vertex in the same inner connected component.
The above follows from the fact that every path in $\{ Q_{\ell}\}_{\ell=1}^M$ contains at most $\varepsilon \log{n}$ edges and from the definition of a distinguished vertex.

For a fixed inner connected component, we define its {\em distinguished graph} whose vertices are the distinguished vertices of the inner connected component, and two such vertices are connected by an edge if and only if they are at most $2\varepsilon \log{n} $ edges away (using only edges in $E_H$, {\em i.e.}, intra cloud edges).
From the above and the fact that the inner connected component is connected it follows that this graph is connected.

Let us take an arbitrary spanning tree of the distinguished graph, and use it to bound the number of different representations. It is enough to describe the structure of the tree and the differences (in the group whose Cayley graph is $H$) corresponding to the tree's edges, since this is all the information appearing in the formal transformation $\{ Q_{\ell}\} _{\ell = 1}^M$.
Every edge in the distinguished graph corresponds to a concatenation of at most $2\varepsilon \log{n}$ edges from the formal transformation $\{ Q_{\ell}\} _{\ell = 1}^M$.
Thus, the difference in the group $H$ is its Cayley graph between the two endpoints of the given edge has at most the following number of options:
$$ \sum _{i=1}^{2\varepsilon \log{n}} d^i \leq d^{O(\varepsilon \log{n})} = n^{O(\varepsilon)}.$$

Let us denote by $p$ the number of distinguished vertices in the distinguished graph.
Every spanning tree of the distinguished graph has exactly $p-1$ edges.
Thus, the total number of possible differences on the edges of any given spanning tree is at most: $n^{O(\varepsilon (p-1))}$.
Recalling that Cayley's formula for counting the number of spanning trees provides an upper bound of $ p^{p-2}$ on the number of spanning trees when there are $p$ vertices present, yields that the number of options for edge labeled spanning trees is at most: $ n^{O(\varepsilon (p-1))}\cdot p^{p-2}$.

Let us now focus on the number of options for choosing the indices $ \text{ind}$ (as provided by Algorithm \ref{alg:FormalTrans}) for the distinguished vertices. 
Lemma \ref{lem:versInCloud} implies that there are at most $n^{O(\varepsilon)}$ vertices in the inner connected component, thus providing the same upper bound on the number of distinguished vertices, $p$.
Therefore, the total number of options for the indices $\text{ind}$ (as provided by Algorithm \ref{alg:FormalTrans}) can be upper bounded by $ (n^{O
(\varepsilon)})^p = n^{O(\varepsilon p)}$.
Thus, we can conclude that the total number of options for a representation, given $p$, is at most: 
\begin{align*}
 n^{O(\varepsilon (p-1))}\cdot p^{p-2}\cdot  n^{O(\varepsilon p)} = & n^{O(\varepsilon p)}\cdot p^{p-2}   \\
 =& n^{O(\varepsilon p)} \cdot n^{O(\varepsilon) (p-2)} \\
 =& n^{O(\varepsilon p)}
\end{align*}

One can note that $p$ ranges from $1$ to at most $s+r$, recalling that $r$ is the number of representatives in the inner connected component and $s$ is the number of surprising edges touching the inner connected component.
Plugging everything together yields an upper bound of: $$ \sum _{p=1}^{s+r}n^{O(\varepsilon p)}\cdot p^{p-2} \leq n^{O(\varepsilon (s+r))}.$$
This concludes the proof.
\end{proof}

\vspace{5pt}
\noindent {\bf{Definition of Certificates.~~~~~}}
We are now ready to define what a certificate is.
To this end we define the {\em skeleton} of a formal transformation $\{ Q_{\ell}\} _{\ell =1}^M$ by going over all paths starting from $\ell=1$ up to $\ell = M$ and removing all indices $\text{ind}$ (as given by Algorithm \ref{alg:FormalTrans}) of the vertices of $Q_{\ell}$, starting with the first vertex and scanning $Q_{\ell}$ towards its end vertex, except for: $(1)$ the first vertex of  $Q_{\ell}$; $(2)$ the last vertex of $Q_{\ell}$; and $(3)$ the target of a surprising edge if it leads to an inner connected component and it is the first occurrence of this edge in the skeleton (with respect to all previous paths scanned).
We denote the skeleton of $\{ Q_{\ell}\} _{\ell =1}^M$ by $\text{skel}(\{ Q_{\ell}\} _{\ell =1}^M)$.

\begin{definition}\label{def:cert}
Given $G_X$ sampled from $\ExtGH$ and a split $\bar{f}$, let us denote by $\{ P_{\ell}\}_{\ell=1}^M$ the collection of all shortest paths in $G_X$ between $\bar{f}(g_1)$ and $\bar{f}(g_2)$ for every $(g_1,g_2)\in E_{G'}$, by $\{ Q_{\ell}\}_{\ell=1}^M$ the formal transformation of $\{ P_{\ell}\}_{\ell=1}^M$, and by $R$ the inner connected component graph obtained from $\{ Q_{\ell}\}_{\ell=1}^M$.
The {\em certificate} of $G_X$ and $\bar{f}$ is a quadruplet that consist of:
\begin{enumerate}
    \item $G'$. \label{cert0}
    \item The representations of all inner connected components of $R$.\label{cert1}
    \item $\text{skel}(\{ Q_{\ell}\}_{\ell=1}^M) $. \label{cert2}
    \item $\{ \bar{f}(g)\}_{g\in V_{G'}}$.\label{cert3}
\end{enumerate}
We denote the above by $\text{cert}(G_X,\bar{f})$.
We call a certificate {\em proper} if it can be obtained from a $G_X$ in the support of $\ExtGH$ and an $(\varepsilon \log ^{\nicefrac[]{4}{3}}n,\varepsilon)$-split $\bar{f}$.
\end{definition}
Two notes regarding the above definition.
First, the reader should recall that $\bar{f}(g)\in X=V_G\times V_H$, for every $g\in V_{G'}$.
Therefore, for every representative $\bar{f}(g)$ the certificate contains both the cloud of the representative and its true identity within the cloud, {\em i.e.}, the element in the group $H$ is its Cayley graph.
Second, the reader should note that given only the fourth ingredient of the certificate (see \ref{cert3} in the above definition) and only the edge labels of $\text{skel}(\{ Q_{\ell}\} _{\ell=1}^M)$ (see \ref{cert2} in the above definition), it is possible to reconstruct the cloud of each vertex in $\text{skel}(\{ Q_{\ell}\} _{\ell=1}^M)$.

Our goal is to show that given a certificate $\text{cert}(G_X,\bar{f})$ one can reconstruct $R$ that corresponds to $G_X$ and $\bar{f}$.
On the other hand, one can easily note that from a given certificate $\text{cert}(G_X,\bar{f})$ there are objects that cannot be reconstructed, {\em e.g.}, $G_X$.
To this end we show the following lemma.
\begin{lemma}\label{lem:Rreconstruction}
Given any certificate $\text{cert}(G_X,\bar{f})$ one can reconstruct the inner connected component graph $R$  that corresponds to $G_X$ and $\bar{f}$.
\end{lemma}
\begin{proof}
We claim that when one scans the paths in $\text{skel}(\{ Q_{\ell}\} _{\ell = 1}^M)$, from $\ell=1$ to $\ell=M$, and each path from its start vertex to its end vertex, then one can always distinguish between the following cases:
$(1)$ the current vertex belongs to an inner connected component; and
$(2)$ the current vertex belongs to a subpath that corresponds to an edge in $E_R$.

Moreover, we know the following things.
For case $(1)$ above, we also know which inner connected component the vertex belongs to, and the quotient between the current vertex and all distinguished nodes of the inner connected component (with respect to the group $H$ is its Cayley graph).
For case $(2)$ above, we also know the starting vertex index $\text{ind}$ (as given by Algorithm \ref{alg:FormalTrans}) of the corresponding subpath and the label of its first edge, which can determine whether we encountered this subpath before, and if we did, which subpath it was and our relative position in this subpath.
If this subpath was not encountered before, once the end vertex of this subpath will be reached and the scan will be back in case $(1)$, we will know the end vertex of the edge in $E_R$ that corresponds to this subpath.

The above follows by induction on the order of the scan and follows from the definition of a certificate and the fact that the collection of all inner connected component in $R$ is exactly the collection of all connected component that have degree at least three or contain a representative.
\end{proof}

\subsection{Counting Number of Certificates}
In this section we aim to upper bound the possible number of certificates $\text{cert}(G_X,\bar{f})$ that can be obtained, for any $G_X$ in the support of $\ExtGH$ and $(\varepsilon \log^{\nicefrac[]{4}{3}}{n},\varepsilon)$-split $\bar{f}$ ,{\em i.e.}, we upper bound the number of proper certificates.
In the following two lemmas, all bounds are given as a function of $s_{\text{tot}}$ (recall that $s_{\text{tot}}$ is the sum of the degrees of the inner connected components graph $R$).
Thus, for any given fixed value of $s_{\text{tot}} $ we upper bound the number of certificates.

\begin{lemma}\label{Lemma:CompDivCount}
There are at most $n^{O(\varepsilon(s
_{\text{tot}}+n))} $ options for choosing the number of inner connected components for all the clouds $g\in V_G$, and for each inner connected component its $s$ (the degree of the inner connected component) and $r$ (the number of representatives in the inner connected component).
\end{lemma}
\begin{proof}
Any option of choosing the number of inner connected components, for all the clouds in the graph $G$, can be represented by a sequence of  $\{ 0,1,2\}$, where:
$0$ indicates a single degree of an inner connected component, $1$ indicates changing to the next cloud, and $2$ indicates  a new inner connected component in the current cloud.
The restrictions on any above sequence are the following:
$(1)$ there are exactly $s_{\text{tot}}$ $0$s in the sequence; $(2)$ there are exactly $n-1$ $1$s in the sequence; and $(3)$ there are at most $s_{\text{tot}}+n$ $2$s (recall that an inner connected component has degree at least three or it contains a representative).
Thus, the number of sequences is at most $\sum _{i=0} ^{s_{\text{tot}}+n} 3^{s_{\text{tot}}+n-1+i} \leq 3^{2s_{\text{tot}}+2n} \leq n^{O(\varepsilon(s
_{\text{tot}}+n))}$.
\end{proof}

\begin{lemma}\label{Lemma:CertCount}
There are at most $n^{n+O(\varepsilon (s_{\text{tot}} + n))}$ proper certificates.
\end{lemma}
\begin{proof}
First, we need to bound the number of graphs $G'$, see \ref{cert0} in Definition \ref{def:cert}.
Since $G'$ is  a subgraph of $G$ which contains $O(n)$ vertices and edges in total there are at most $2^{O(n)}\leq n^{O(\varepsilon n)}$ such graphs (here we used that $2\leq n^{O(\varepsilon)} $ for large enough $n$).
Second, let us focus on the number of options for the second ingredient of a certificate, see \ref{cert1} in Definition \ref{def:cert}.
For any choice of $s$ and $r$ for all inner connected components and the clouds they belong to, Lemma \ref{Lemma:CompRepCount} provides an upper bound that is of the form: $ N(s_1,r_1)\cdot N(s_2,r_2)\cdot \ldots$, given a list of $s_i$s and $r_i$s and to which cloud each pair of $s_i$ and $r_i$ refers to.
This upper bound is: $$ N(s_1,r_1)\cdot N(s_2,r_2)\cdot \ldots \leq n^{O(\varepsilon (\sum _i {s_i}+\sum _i {r_i}))} \leq n^{O(\varepsilon (s_{\text{tot}}+n))}.$$
The last inequality follows from the fact that there are at most $n$ representatives, {\em i.e.}, the sum of the $r_i$s is at most $n$.
Thus, applying Lemma \ref{Lemma:CompDivCount} provides a following total upper bound on all inner connected components:
$$ n^{O(\varepsilon (s_{\text{tot}}+n))}\cdot n^{O(\varepsilon (s_{\text{tot}}+n))} \leq  n^{O(\varepsilon (s_{\text{tot}}+n))} .$$
This upper bounds the number of options for the second ingredient of a certificate (\ref{cert1} in Definition \ref{def:cert}).

Let us now bound the number of options for the fourth ingredient in the definition of a certificate, see \ref{cert3} in Definition \ref{def:cert}.
The third requirement in Definition \ref{def:Split} implies that there are at most $n^{O(\varepsilon)}$ options for choosing $\pi \circ \bar{f}(g)$, for every $g\in V_{G'}$.
This implies that for every $g\in V_{G'}$, the number of options for choosing $\bar{f}(g)$ is at most $n^{1+O(\varepsilon)}$.
Since there are at most $n$ vertices in $V_{G'}$, the total number of options for the fourth ingredient of a certificate is upper bounded by $n^{n+O(\varepsilon n)}$.

Finally, let us bound the number of possible $ \text{skel}(\{ Q_{\ell}\}_{\ell=1}^M)$, the third ingredient of a certificate (see \ref{cert2} in Definition \ref{def:cert}).
Fix and edge $(g_1,g_2)\in E_{G'}$, which corresponds to a path $P_\ell$ in $G_X$ between $\bar{f}(g_1)$ and $\bar{f}(g_2)$ that contains at most $\varepsilon \log{n}$ edges (recall the second requirement of Definition \ref{def:Split} and the fact that every edge in $E_X$ has length at least $\log ^{\nicefrac[]{1}{3}}n$).
Each edge in the skeleton obtained from the formal transformation $Q_{\ell} $ of $P_{\ell}$ has $2d$ options for a label in $\LX$, thus a total of $ \sum _{i=0}^{\varepsilon \log{n}} (2d)^i \leq d^{O(\varepsilon \log{n})} = n^{O(\varepsilon)}$ for the edge labels of the skeleton of a single $Q_{\ell}$ and $n^{O(\varepsilon n)}$ options for all $\{ Q_{\ell}\}_{\ell=1}^M$.

All that remains is to bound the number of options for the indices $\text{ind}$ (of Algorithm \ref{alg:FormalTrans}) of the vertices that should have an index $\text{ind}$ (according to the definition of a skeleton). Recall that this index $\text{ind}$ contains the cloud and a running number.
It is important to note that given the fourth ingredient of the certificate (which specifies the clouds of the representatives), which we already counted, and the edge labels of $\text{skel}(\{ Q_{\ell}\} _{\ell=1}^M)$, which we counted as well, it is possible to reconstruct the cloud of each vertex in $\text{skel}(\{ Q_{\ell}\} _{\ell=1}^M)$.
Hence, we only need to bound the number of options for assigning the running index of $\text{ind}$ (which we denoted by $i$).

There are at most $s_{\text{tot}}$ vertices that are the target of a surprising edge.
Additionally, there are at most $n$ vertices that are the start or end vertex of a path in $\{ Q_{\ell}\}_{\ell=1}^M$ (as they are the representatives).
Thus, the total number of vertices for which index $\text{ind}$ (of Algorithm \ref{alg:FormalTrans}) should appear in $\text{skel}(\{ Q_{\ell}\}_{\ell=1}^M)$ is at most $s_{\text{tot}}+n$.
We need to bound the number of options of selecting these vertices, and for each such selection the number of options for the running indices.
As there are at most $n \varepsilon \log(n)$ vertices in the skeleton on total, there are at most $2^{n\varepsilon \log{n}}=n^{O(\varepsilon n)}$ options for choosing these vertices (even though not all of these options may have the right number of distinguished vertices).
Recall that from lemma \ref{lem:versInCloud} these running indices go up to $n^{O(\varepsilon)}$. Hence, we can conclude that the number of options for determining the running indices is at most $n^{O(\varepsilon (n+s_{\text{tot}}))}$.

The proof is concluded by aggregating all four components of the Definition \ref{def:cert}.
\end{proof}

\subsection{Bounding Probability of a Certificate}
In this section we bound the probability, over the random choice of $G_X$ from $\ExtGH$, of obtaining a given fixed proper certificate.

\vspace{5pt}
\noindent {\bf{Interlude in Linear Algebra.~~~~~}}
Recall that the notation $\mathbb{F}_2^A$ denotes the vector space whose coordinates are indexed by $A$ over $\mathbb{F}_2$.
Given a graph $F=(V_F,E_F)$, we denote by $\delta _{F}(v)$ the collection of edges in $ F$ that touch $v$.
We define the following linear transformation $\eta_{F}:\mathbb{F}_2^{E_{F}}\rightarrow \mathbb{F}_2^{V_{F}}$ as follows:
$\left(\eta _{F}(\mathbf{x})\right) _v\triangleq \sum _{e\in \delta _{F}(v)}x_e$, for every $v\in V_{F}$ and every $ \mathbf{x}\in \mathbb{F}_2^{E_{F}}$.
The following lemma is well known, one can refer to, {\em e.g.}, Diestel \cite{D05} Lemma $1.9.6$.
\begin{lemma}\label{Lemma:euler}
The following holds:
$$dim \ker{\eta_{F}} = \lvert E_F \rvert - \lvert V_F \rvert + \#(\text{connected components of F}).$$
\end{lemma}
From now on, we will denote $dim \ker{\eta_{F}}$ as $b_1(F)$ (note that this is the rank of the first homology with coefficients in $ \mathbb{Z}_2$ which are analogues Betti numbers).

\vspace{5pt}
\noindent {\bf{Surprising Edges and their Probabilities.~~~~~}}
Recall that $s_{\text{tot}}$ is the sum of the degrees of inner connected components.
\begin{lemma}\label{Lemma:stotICC}
$\frac{1}{6} s_{tot} - \frac{1}{3}n \leq b_1(R) \leq \frac{1}{2}s_{tot}$.
\end{lemma}
\begin{proof}
We count the number of edges of $R$ in two different ways.
There are two types of vertices in $R$: vertices which have at least one representative and vertices which have no representative but are of degree at least $3$.
As the number of representatives is at most $n$, we can conclude that $s_{tot} \geq 3\left(|V_R|-n\right)$. As the sum of degrees is equal to twice the number of edges we have $s_{tot} = 2 |E_R|$.
Clearly, the number of connected components of $G_R$ is at most $|V_R|$.
Combining these equations we get:
$\frac{1}{2}s_{tot} = |E_R| \geq |E_R|-(|V_R|-\# (\text{connected components of }G_R)) \geq |E_R|-|V_R| \geq \frac{1}{2}s_{tot} - \frac{1}{3}s_{tot} + \frac{1}{3}n = \frac{1}{6} s_{tot} - \frac{1}{3}n$, where last inequality follows from the above two bounds.
The proof is concluded by recalling that $b_1(R)=|E_R|-(|V_R|-\# (\text{connected components of }G_R))$, which follows from Lemma \ref{Lemma:euler}.
\end{proof}

The following lemma provides an upper bound on the probability of obtaining a given proper certificate, as a function of $b_1(R)$.
One should recall that given a proper certificate, one can reconstruct $R$ (Lemma \ref{lem:Rreconstruction}).
The following lemma is the only place in our analysis where the fact that the matchings in $G_X$ are chosen uniformly at random is used.

\begin{lemma}\label{Lemma:probBound}
Given a proper certificate the probability, over the random choice of $G_X$, of obtaining it is at most $(2/n)^{b_1(R)}$.
\end{lemma}

\begin{proof}
We define a process of scanning $R$ whose result is a collection of inter-cloud edges in $G_X$.
We denote this collection by {\em constraint edges}.
Our proof is based on upper bounding the probability that these chosen constraint edges are obtained when sampling $G_X$.

First, let us describe the scanning process of $R$.
We start from an arbitrary vertex in $V_R$, and perform a DFS algorithm of $R$.
During this DFS scan, we expose the element in the group $H$ is its Cayley graph, {\em i.e.}, the true identity, of all vertices in the certificate (the identity $g$ of the cloud is already known).
This is performed as follows:
\begin{enumerate}
    \item Every time we enter a new inner connected component for the very first time we expose the true identity of all vertices in the inner connected component. This can be achieved by knowing the identity of a single vertex in the inner connected component, {\em e.g.}, the vertex we entered with, and then inferring all other identities using the certificate which contains the representation of the inner connected component.
    \item Every time we traverse an edge in  $E_R$ we expose the true identities of all vertices in its corresponding path in $G_X$. If the final vertex of the path is a vertex whom we already know its identity, then we add the matching edge leading to it to the collection of constraint edges.
\end{enumerate}
It is important to note that the above DFS scan might expose some partial information on the random matching between two neighboring clouds, {\em  i.e.}, when the DFS traverses an inter-cloud edge in $G_X$ we expose the information that the true identities of its two endpoints are matched.
Note that since the certificate is proper, the information exposed from the DFS scan can never be in a contradicting state with the certificate, {\em e.g.}, two different vertices in the same cloud have an identical true identity.

Each constraint edge $((g,h),(g',h'))$ gives a constraint on the uniform random matching between the clouds $g$ and $g'$.
This constraint is of the form: $(g,h)$ is matched to $(g',h')$.
We partition the constraint edges according to the matching they belong to.
For every matching $M_{(g,g')}$ (where $(g,g')\in E_G$), recall that $M_{(g,g')}$ is a uniform random matching in a complete bipartite graph with $V_H$ on the left and $V_H$ on the right.
We denote by $ \alpha_{(g,g')}$ the number of edges in the certificate that belong to $M_{(g,g')}$ and are not constraint edges, and by $\beta_{(g,g')}$ the number of constraint edges in $M_{(g,g')}$.
In what follows we condition on any (non-zero probability) event that dictates the true identities of the endpoints of all $\alpha _{(g,g')}$ matching edges that are in the certificate but are not constraint edges. 
We note that under the above conditioning, the probability of the $ \beta _{(g,g')}$ constraint edges of $M_{(g,g')}$ to correspond to the exposed true identities of its endpoints equals: $$\prod_{i=1}^{\beta_{(g,g')}} \frac{1}{n - \alpha_{(g, g')} - i +  1} .$$
We note that:
\begin{align}
 \prod_{i=1}^{\beta_{(g,g')}} \frac{1}{n - \alpha_{(g, g')} - i +  1} \le \prod_{i=1}^{\beta_{(g,g')}} \frac{1}{n - \alpha_{(g, g') }- \beta_{(g,g')}} \le \left(\frac{2}{n}\right)^{\beta_{(g,g')}} . \label{inequality}  
\end{align}
The last inequality follows for a large enough $n$ and from Lemma \ref{lem:versInCloud} which implies that $\alpha_{(g, g')} + \beta_{(g,g')} \le n^{O(\varepsilon)}$.

Similarly to the above, let us now condition on any (non-zero probability) event that dictates the true identities of the endpoints of all $\alpha _{(g,g')}$ matching edges that are in the certificate but are not constraint edges and for all matchings $(g,g')\in E_G$.
The event that we wish to upper bound its probability is that we obtain the given fixed proper certificate.
We note that the above event that we condition on, implies the true identities of all vertices in the certificate, and thus for every constraint edge in every matching we obtain the desired true identities of its endpoints.
The probability of the event that every constraint edge in every matching corresponds to the desired true identities of its endpoints, upper bounds the probability of obtaining the given fixed proper certificate.
Thus, since all matchings are chosen independently, we can multiply (\ref{inequality}) over all matchings and obtain the following upper bound: 
$(2/n)^{\sum _{(g,g')\in E_G} \beta_{(g,g')}} $.
Since this bound does not depend on the event we conditioned on, we can use the law of total probability and conclude that this upper bound holds unconditionally.

To conclude the proof we will prove that $\sum _{(g,g')\in E_G}\beta_{(g,g')}= b_1(R)$. Removing from $R$ all edges whose corresponding path  in $G_X$ contains a constraint edge, leaves a spanning tree for each connected component of $R$. Thus, the number of remaining edges is $|V_R|$ minus the number of the connected component of $R$, which equals by Lemma \ref{Lemma:euler} to $|E_R| - b_1(R)$.
On the other hand this number equals also to $|E_R| - \sum _{(g,g')\in E_G}\beta_{(g,g')}$, as there is at most a single constraint edge in every path in $G_X$ that corresponds to an edge in $R$.
This equality concludes the proof.
\end{proof}

\vspace{5pt}
\noindent {\bf{Putting it All Together.~~~~~}}
Let us finalize our analysis, by providing a lower bound on $b_1(R)$.
This will result in an absolute upper bound on the probability of obtaining a given proper certificate.
This is captured by the following lemma, which is the only place in our proof that we use the definition of cycle-homeomorphism.

\begin{lemma}\label{Lemma:BettiBound}
For every certificate $b_1(R) \geq \left(1-O(\varepsilon)\right)\left(\frac{d}{2} - 1\right)\cdot n $.
\end{lemma}
\begin{proof}
In the proof we focus on $\eta _{R}$, $\eta _G$, and $\eta _{G'}$, and how they interact.
The heart of the proof is that the size of the kernel of $\eta _{R}$ is at least as large as the size of the kernel of $ \eta _{G'}$.

The proof requires the following additional linear transformation $ \pi :\mathbb{F}_2^{E_{R}}\rightarrow \mathbb{F}_2^{E_G}$.
We define $\pi$ by how it operates on the basis of $\mathbb{F}_2^{E_{R}}$.
Note that each base element of $\mathbb{F}_2^{E_{R}}$ corresponds to an edge $e\in E_{R} $, which in turn
corresponds to a subpath in $\text{skel}(\{ Q_{\ell}\} _{\ell = 1}^M)$.
Thus, $e\in E_{R}$ defines a path in $G$ by the natural projection that assigns every $(g,h)$ to $g$. 
Therefore, $\pi (\mathbf{1}_e)\in \mathbb{F}_2^{E_G}$ equals in its $e'\in E_G$ coordinate the parity of the number of occurrences of $e'$ in the above defined path.\footnote{Note that if the path in $G$ that corresponds to $e\in E_{R}$ is not simple, then $\pi (\mathbf{1}_e)$ might equal a path and and an additional collection of disjoint cycles.}

We can also define $\tilde{\pi}:\mathbb{F}_2^{V_{R}}\rightarrow \mathbb{F}_2^{V_G}$ by defining how it operates on the basis of $\mathbb{F}_2^{V_{R}}$: $\tilde{\pi}(\mathbf{1}_v)$ equals $1$ in the coordinate that corresponds to the cloud the inner connected component $v$ resides in and $0$ otherwise.
A crucial, yet simple, observation is that $\tilde{\pi}(\eta _{R}(\mathbf{x})) = \eta _G(\pi (\mathbf{x}))$, for every $\mathbf{x}\in \mathbb{F}_2^{E_{R}}$.
Thus, a consequence of this crucial observation is that $\sigma$ sends elements in the kernel of $\eta _{G'}$ to elements in the kernel of $\eta _{R}$.

Similarly to the definitions of $\pi$ and $\tilde{\pi}$, we can define $\sigma:\mathbb{F}_2^{E_{G'}}\rightarrow \mathbb{F}_2^{E_R} $ and $\tilde{\sigma}:\mathbb{F}_2^{V_{G'}}\rightarrow \mathbb{F}_2^{V_{R}} $.
For each base element of $e '\in \mathbb{F}_2^{E_{\text{G'}}} $, we take the path in $R$ between the two inner connected components that contain the representatives of the two endpoints of $e'$ that originated from the appropriate $Q_{\ell}$.
Thus, we can define $\sigma(e')$ as the sum in $\mathbb{F}_2^{E_{R}}$ over edges appearing in this path.
Furthermore, for each base element $v\in V_{G'}$ we can define $\tilde{\sigma}(\mathbf{1}_v)$ to be the inner connected component, {\em i.e.}, a vertex in $V_R$, containing the representative of $v$. 
As before, one can note that $ \tilde{\sigma}(\eta _{\text{G'}}(\mathbf{x})) = \eta _{R}(\sigma (\mathbf{x}))$, for every $\mathbf{x}\in \mathbb{F}_2^{E_{G'}}$.
The above transformations can be captured in the following commutative diagram:

\[\begin{tikzcd}
& \mathbb{F}_2^{E_{G'}} \ar[r, "\eta_{G'}"] \ar[d, "\sigma"] & \mathbb{F}_2^{V_{G'}}  \ar[d, "\tilde{\sigma}"]\\
& \mathbb{F}_2^{E_{R}} \ar[r, "\eta_{R}"] \ar[d, "\pi"] & \mathbb{F}_2^{V_{R}} \ar[d, "\tilde{\pi}"]\\
& \mathbb{F}_2^{E_{G}} \ar[r, "\eta_{G}"] & \mathbb{F}_2^{V_{G}}\\
\end{tikzcd}\]

Note that if we prove that $dim\ker \eta _{R} \geq dim \ker \eta _{G'}$ then the proof is concluded.
Lemma \ref{Lemma:euler} implies that:
 $$ dim \ker \eta _{G'}= |E'|-|V'|+\text{\#(connected components of $G'$)}.$$ 
 We note that: $ |E'|-|V'|+\text{\#(connected components of $G'$)} \geq (1-O(\varepsilon)) (d/2-1)n$.
This inequality is true since $|E'|\geq (1-O(\varepsilon))|E|$ (see Definition \ref{def:Split}), $|V'|\leq |V|$, and the number of connected components is non negative.
Thus:
\begin{align*}
 dim \ker \eta _{G'} & \geq (1-O(\varepsilon))|E| - |V| \\   
 & \geq (1-O(\varepsilon))\frac{d}{2} \cdot n - n\\
 & = \left(1-O(\varepsilon)\right)\left(\frac{d}{2} - 1\right)\cdot n .
\end{align*}

Given the above inequality, let us now focus on proving that $dim\ker \eta _{R} \geq dim \ker \eta _{G'}$.
In order to do that, it suffices to show that $\sigma$ is injective when restricted to $\ker \eta _{\text{G'}}$.
As the certificate corresponds to a split which is a cycle-homomorphism (see Definitions \ref{Def:CycleHomomorphism} and \ref{def:Split}), we claim that for any given cycle $C\subseteq E_{G'}$: $ \pi \circ \sigma (\mathbf{1}_C)=\mathbf{1}_C$.
This follows directly from the definition of cycle-homeomorphism.
We note that every element in $\ker \eta _{G'}$ can be written as a sum of simple cycles.
Thus, $\pi \circ \sigma$ coincides with the natural injection for all elements in $\ker \eta_{G'}$ and in particular $\pi \circ \sigma$ is injective restricted to $\ker \eta_{G'}$.
Hence, $\sigma$ is injective restricted to $\ker \eta_{G'}$, which in turn implies that $dim\ker \eta _{R} \geq dim \ker \eta _{G'}$.
\end{proof}

\subsection{Proof of Theorem \ref{thrm:NoSplit}}
\begin{proof}[Proof of Theorem \ref{thrm:NoSplit}]
Let us prove that there is a graph $G_X$ that can be sampled from $\ExtGH$ that has no proper certificate.
We start by fixing the value of $b_1(R)$ and asking how many proper certificates can attain this fixed value (recall that given a certificate one can reconstruct $R$ from it by Lemma \ref{lem:Rreconstruction}).
From Lemma \ref{Lemma:BettiBound} we can deduce that $b_1(R) = \Omega(n)$, thus Lemma \ref{Lemma:stotICC} implies $b_1(R) = \Theta(s_{tot})$.
Thus, fixing a value of $b_1(R)$, gives us $O(b_1(R))$ options for the value of $s_{tot}$ that cannot exceed $O(b_1(R))$.
Using Lemma \ref{Lemma:CertCount}, we can infer that there are at most
$\sum _{s_{\text{tot}}=0}^{O(b_1(R))} n^{n+O(\varepsilon (n + s_{\text{tot}}))} \leq n^{n+O(\varepsilon b_1(R))}$ options for a proper certificate with the given $b_1(R)$.

Now let us consider the probability of obtaining a given proper certificate.
Lemma \ref{Lemma:probBound} upper bounds the probability of obtaining a given proper certificate by $(2/n)^{b_1(R)}$.
Recalling that $b_1(R) = \Omega(n)$ and $2 = o(n^\varepsilon)$ (since $\varepsilon$ is a constant) we can upper bound $(2/n)^{b_1(R)}$ by  $n^{-(1-\varepsilon) b_1(R)}$.
Lemma \ref{Lemma:BettiBound} gives a lower bound on $b_1(R)$, let's denote it by $B_1$.
Combining these two bound via a simple union bound, one can obtain that the probability there exists an proper certificate that can be obtained from a graph $G_X$ sampled from $\ExtGH$ is at most:
\begin{align*}
 \sum_{b_1(R)=B_1}^{\infty} n^{n+O(\varepsilon b_1(R))} \cdot n^{-(1-\varepsilon)b_1(R)} = & \sum_{b_1(R)=B_1}^{\infty} n^{n - (1-O(\varepsilon))b_1(R)}   \\
 \leq & n^{n - (1-O(\varepsilon))B_1} .
\end{align*}
The last inequality follows from the facts that: $(1)$ a geometric series with a quotient which is less than half sums up to at most twice the first element in the series; and $(2)$ $2 = o(n^{\varepsilon B_1})$.
Thus, to conclude the proof it suffices to prove that $ n^{n - (1-O(\varepsilon))B_1} < 1$.
The latter is equivalent to $n < (1-O(\varepsilon))B_1 $.
The proof is concluded since for every $d\geq 3$ and small enough constant $\varepsilon $ the following is true:
$ (1-O(\varepsilon))B_1= \left(1-O(\varepsilon)\right)\left(\frac{d}{2} - 1\right)\cdot n$.
\end{proof}

\section{Discussion and Future Research}
In this section we discuss some aspects of our construction, as well as future research.

First, let us focus on our construction.
Regarding the \textbf{use of Cayley graphs}, we believe that it is superfluous to assume that $G$ and $H$ are Cayley graphs. We use this extra structure to reduce the information needed describing the inner connected components and thus strengthening our bound on the number of certificates. Alternatively, instead of assuming that $G$ and $H$ are Cayley graphs, one can add the identity (in $H$) of one of the vertices in the component. This introduces problems when considering components of degree $3$.
    However, such components have less vertices than the girth of $H$ and thus have a simple combinatorial structure which can be counted separately.
    
Considering our \textbf{parameters selection}, there are $3$ main parameters that need to be chosen: length of intra-cloud edges, length of inter-cloud edges, and length of edges connecting a terminal to its neighboring non-terminal vertex.
    We note that our choice of parameters is optimal for this construction.
    This can be seen as there are three natural integral solutions, each provides a different lower bound on the value of the integral solution.
    The first assigns all non-terminals to the same terminal, the second assigns every non-terminal to its (single) neighboring terminal, and the third assigns all vertices in a cloud to a random terminal inside the cloud.
    The parameters are chosen as to balance the bounds provided by the above solutions.

Second, let us briefly mention specific future research direction.
We conjecture that \textbf{recursively repeating} our randomized graph extension with $t$ graphs (with appropriate edge lengths) will give an integrality gap of $\Omega(\log^{\nicefrac{t}{(t+1)}}{k})$.
    Analyzing this more general construction remains an open question.
When considering the \textbf{earthmover relaxation}, in order to provide an integrality gap for this relaxation via our approach, one needs to adapt Section \ref{sec:IntegralSol}.
    Specifically, one needs to present an instance for which the existence of a cheap integral solution implies a split.
\section*{Acknowledgments}
The authors would like to thank Yuval Rabani for sharing with them the algorithm and its proof that when $D$ is the shortest path metric of a high girth expander {\EXT} admits an approximation of $O(1)$.
The authors are also grateful to Yuval Filmus and Prahladh Harsha for pointing to them the relevant literature relating to lifts of graphs.
Moreover, the authors would like to thank Roy Meshulam for insightful discussions.
Finally, the authors would like to thank the anonymous reviewers for helpful remarks regarding the presentation of the paper.

This research was supported by European Horizon 2020 research programm under grant agreement 852870, and ISF 1336/16.
\clearpage

\bibliographystyle{alpha}
\bibliography{references}

\newcommand{\etalchar}[1]{$^{#1}$}
\begin{thebibliography}{KKMR09}

\bibitem[AFH{\etalchar{+}}04]{AFHKTT04}
Aaron Archer, Jittat Fakcharoenphol, Chris Harrelson, Robert Krauthgamer, Kunal
  Talwar, and \'{E}va Tardos.
\newblock Approximate classification via earthmover metrics.
\newblock In {\em Proceedings of the Fifteenth Annual ACM-SIAM Symposium on
  Discrete Algorithms}, SODA '04, page 1079–1087, 2004.

\bibitem[AKK99]{AKK99}
Sanjeev Arora, David Karger, and Marek Karpinski.
\newblock Polynomial time approximation schemes for dense instances of np-hard
  problems.
\newblock {\em Journal of Computer and System Sciences}, 58(1):193 -- 210,
  1999.

\bibitem[AKK{\etalchar{+}}08]{AKKSTV08}
Sanjeev Arora, Subhash~A. Khot, Alexandra Kolla, David Steurer, Madhur
  Tulsiani, and Nisheeth~K. Vishnoi.
\newblock Unique games on expanding constraint graphs are easy: Extended
  abstract.
\newblock In {\em Proceedings of the Fortieth Annual ACM Symposium on Theory of
  Computing}, STOC '08, page 21–28, 2008.

\bibitem[AKM13]{AKM13}
Naman Agarwal, Alexandra Kolla, and Vivek Madan.
\newblock Small lifts of expander graphs are expanding.
\newblock {\em ArXiv}, abs/1311.3268, 2013.

\bibitem[AMM17]{AMM17}
Haris Angelidakis, Yury Makarychev, and Pasin Manurangsi.
\newblock An improved integrality gap for the
  {C}{\u{a}}linescu-{K}arloff-{R}abani relaxation for multiway cut.
\newblock In {\em Integer Programming and Combinatorial Optimization}, pages
  39--50. Springer International Publishing, 2017.

\bibitem[{Bar}96]{B96}
Y.~{Bartal}.
\newblock Probabilistic approximation of metric spaces and its algorithmic
  applications.
\newblock In {\em Proceedings of 37th Conference on Foundations of Computer
  Science}, FOCS '96, pages 184--193, 1996.

\bibitem[Bar98]{B98}
Yair Bartal.
\newblock On approximating arbitrary metrices by tree metrics.
\newblock In {\em Proceedings of the Thirtieth Annual ACM Symposium on Theory
  of Computing}, STOC '98, page 161–168, 1998.

\bibitem[BCKM19]{BCKM19}
Krist{\'o}f B{\'e}rczi, Karthekeyan Chandrasekaran, Tam{\'a}s Kir{\'a}ly, and
  Vivek Madan.
\newblock Improving the integrality gap for multiway cut.
\newblock In {\em Integer Programming and Combinatorial Optimization}, pages
  115--127. Springer International Publishing, 2019.

\bibitem[BL06]{BL06}
Yonatan Bilu and Nathan Linial.
\newblock Lifts, discrepancy and nearly optimal spectral gap*.
\newblock {\em Combinatorica}, 26:495--519, 10 2006.

\bibitem[BNS18]{BNS18}
Niv Buchbinder, Joseph~(Seffi) Naor, and Roy Schwartz.
\newblock Simplex partitioning via exponential clocks and the multiway-cut
  problem.
\newblock {\em SIAM Journal on Computing}, 47:1463--1482, 01 2018.

\bibitem[BSW17]{BSW17}
Niv Buchbinder, Roy Schwartz, and Baruch Weizman.
\newblock Simplex transformations and the multiway cut problem.
\newblock In {\em Proceedings of the Twenty-Eighth Annual ACM-SIAM Symposium on
  Discrete Algorithms}, SODA '17, page 2400–2410, 2017.

\bibitem[BSW19]{BSW19}
Niv Buchbinder, Roy Schwartz, and Baruch Weizman.
\newblock A simple algorithm for the multiway cut problem.
\newblock {\em Operations Research Letters}, 47(6):587 -- 593, 2019.

\bibitem[CCT06]{CT06}
Kevin Cheung, William Cunningham, and Lawrence Tang.
\newblock Optimal 3-terminal cuts and linear programming.
\newblock {\em Math. Program.}, 106:1--23, 05 2006.

\bibitem[CKNZ04]{CKNZ04}
Chandra Chekuri, Sanjeev Khanna, Joseph Naor, and Leonid Zosin.
\newblock A linear programming formulation and approximation algorithms for the
  metric labeling problem.
\newblock {\em SIAM J. Discrete Math.}, 18(3):608--625, 2004.

\bibitem[CKR00]{CKR00}
Gruia C{\u{a}}linescu, Howard~J. Karloff, and Yuval Rabani.
\newblock An improved approximation algorithm for multiway cut.
\newblock {\em J. Comput. Syst. Sci.}, 60(3):564--574, 2000.

\bibitem[CKR05]{CKR05}
Gruia C{\u{a}}linescu, Howard Karloff, and Yuval Rabani.
\newblock Approximation algorithms for the 0-extension problem.
\newblock {\em SIAM Journal on Computing}, 34(2):358--372, 2005.

\bibitem[CN07]{CN07}
Julia Chuzhoy and Joseph~(Seffi) Naor.
\newblock The hardness of metric labeling.
\newblock {\em SIAM Journal on Computing}, 36(5):1376--1386, 2007.

\bibitem[Die05]{D05}
Reinhard Diestel.
\newblock {\em Graph Theory (Graduate Texts in Mathematics)}.
\newblock Springer, August 2005.

\bibitem[DJP{\etalchar{+}}94]{DJPSY94}
E.~Dahlhaus, D.~S. Johnson, C.~H. Papadimitriou, P.~D. Seymour, and
  M.~Yannakakis.
\newblock The complexity of multiterminal cuts.
\newblock {\em SIAM Journal on Computing}, 23:864--894, 1994.

\bibitem[FHRT03]{FHRT03}
Jittat Fakcharoenphol, Chris Harrelson, Satish Rao, and Kunal Talwar.
\newblock An improved approximation algorithm for the 0-extension problem.
\newblock In {\em Symposium on Discrete Algorithms}, SODA '03, page 257–265,
  2003.

\bibitem[FK96]{FK96}
Alan~M. Frieze and Ravi Kannan.
\newblock The regularity lemma and approximation schemes for dense problems.
\newblock In {\em Proceedings of 37th Conference on Foundations of Computer
  Science}, FOCS '96, pages 12--20, 1996.

\bibitem[FK00]{FK00}
Ari Freund and Howard~J. Karloff.
\newblock A lower bound of 8/(7+(1/k-1)) on the integrality ratio of the
  \mbox{C}alinescu-\mbox{K}arloff-\mbox{R}abani relaxation for multiway cut.
\newblock {\em Inf. Process. Lett.}, 75(1-2):43--50, 2000.

\bibitem[Fri03]{F03}
J.~Friedman.
\newblock Relative expanders or weakly relatively ramanujan graphs.
\newblock {\em Duke Mathematical Journal}, 118:19--35, 2003.

\bibitem[FRT07]{FRT07}
Jittat Fakcharoenphol, Satish Rao, and Kunal Talwar.
\newblock A tight bound on approximating arbitrary metrics by tree metrics.
\newblock {\em Journal of Computer and System Sciences}, 69(3):485--497,
  October 2007.

\bibitem[GT00]{GT00}
Anupam Gupta and \'{E}va Tardos.
\newblock A constant factor approximation algorithm for a class of
  classification problems.
\newblock In {\em Proceedings of the Thirty-Second Annual ACM Symposium on
  Theory of Computing}, STOC '00, page 652–658, 2000.

\bibitem[GVY04]{GVY04}
Naveen Garg, Vijay~V. Vazirani, and Mihalis Yannakakis.
\newblock Multiway cuts in node weighted graphs.
\newblock {\em J. Algorithms}, 50(1):49--61, 2004.

\bibitem[JLS86]{JLS86}
William Johnson, Joram Lindenstrauss, and Gideon Schechtman.
\newblock Extensions of lipschitz maps into banach spaces.
\newblock {\em Israel Journal of Mathematics}, 54:129--138, 06 1986.

\bibitem[Kar98]{K98}
Alexander~V. Karzanov.
\newblock Minimum 0-extensions of graph metrics.
\newblock {\em European Journal of Combinatorics}, 19(1):71 -- 101, 1998.

\bibitem[KKMR09]{KKMR09}
Howard Karloff, Subhash Khot, Aranyak Mehta, and Yuval Rabani.
\newblock On earthmover distance, metric labeling, and 0-extension.
\newblock {\em SIAM J. Comput.}, 39:371--387, 01 2009.

\bibitem[KKS{\etalchar{+}}04]{KKSTY04}
David~R. Karger, Philip~N. Klein, Clifford Stein, Mikkel Thorup, and Neal~E.
  Young.
\newblock Rounding algorithms for a geometric embedding of minimum multiway
  cut.
\newblock {\em Math. Oper. Res.}, 29(3):436--461, 2004.

\bibitem[KLMN05]{KLMN05}
R.~Krauthgamer, J.~R. Lee, M.~Mendel, and A.~Naor.
\newblock Measured descent: {A} new embedding method for finite metrics.
\newblock {\em Geometric And Functional Analysis}, 15(4):839--858, 2005.

\bibitem[KR08]{KR08}
Subhash Khot and Oded Regev.
\newblock Vertex cover might be hard to approximate to within 2-$\epsilon$.
\newblock {\em J. Comput. Syst. Sci.}, 74(3):335--349, May 2008.

\bibitem[KT02]{KT02}
Jon Kleinberg and \'{E}va Tardos.
\newblock Approximation algorithms for classification problems with pairwise
  relationships: metric labeling and markov random fields.
\newblock {\em J. ACM}, 49(5):616--639, 2002.

\bibitem[LN04]{LN04}
James Lee and Assaf Naor.
\newblock Extending lipschitz functions via random metric partitions.
\newblock {\em Inventiones mathematicae}, 160, 02 2004.

\bibitem[LP10]{LP10}
N.~Linial and Doron Puder.
\newblock Word maps and spectra of random graph lifts.
\newblock {\em Random Struct. Algorithms}, 37:100--135, 2010.

\bibitem[LPS88]{LPS88}
Alexander Lubotzky, Ralph Phillips, and Peter Sarnak.
\newblock Ramanujan graphs.
\newblock {\em Combinatorica}, 8:261--277, 09 1988.

\bibitem[MNRS08]{MNRS08}
Rajsekar Manokaran, Joseph~(Seffi) Naor, Prasad Raghavendra, and Roy Schwartz.
\newblock Sdp gaps and ugc hardness for multiway cut, 0-extension, and metric
  labeling.
\newblock In {\em Proceedings of the Fortieth Annual ACM Symposium on Theory of
  Computing}, STOC '08, page 11–20, 2008.

\bibitem[MO20]{MO20}
Sidhanth Mohanty and Ryan O'Donnell.
\newblock \emph{X}-ramanujan graphs.
\newblock In {\em Proceedings of the 2020 {ACM-SIAM} Symposium on Discrete
  Algorithms}, SODA '20, pages 1226--1243, 2020.

\bibitem[MOP20]{MOP20}
Sidhanth Mohanty, Ryan O'Donnell, and Pedro Paredes.
\newblock Explicit near-ramanujan graphs of every degree.
\newblock In {\em Proccedings of the 52nd Annual {ACM} {SIGACT} Symposium on
  Theory of Computing}, STOC '20, pages 510--523, 2020.

\bibitem[MSS13]{MSS13}
A.~{Marcus}, D.~A. {Spielman}, and N.~{Srivastava}.
\newblock Interlacing families i: Bipartite ramanujan graphs of all degrees.
\newblock In {\em 2013 IEEE 54th Annual Symposium on Foundations of Computer
  Science}, FOCS '13, pages 529--537, 2013.

\bibitem[NZ01]{NZ01}
Joseph Naor and Leonid Zosin.
\newblock A 2-approximation algorithm for the directed multiway cut problem.
\newblock {\em SIAM J. Comput.}, 31(2):477--482, 2001.

\bibitem[OW20]{OW20}
Ryan O'Donnell and Xinyu Wu.
\newblock Explicit near-fully x-ramanujan graphs, 2020.

\bibitem[RSW06]{RSW06}
Eyal Rozenman, Aner Shalev, and Avi Wigderson.
\newblock Iterative construction of cayley expander graphs.
\newblock {\em Theory of Computing}, 2:91--120, 01 2006.

\bibitem[SV14]{SV13}
Ankit Sharma and Jan Vondr{\'{a}}k.
\newblock Multiway cut, pairwise realizable distributions, and descending
  thresholds.
\newblock In {\em Symposium on Theory of Computing, {STOC} 2014}, pages
  724--733, 2014.

\end{thebibliography}

\end{document}